\documentclass[manuscript=article,journal=jctcce,layout=twocolumn]{achemso}
\SectionNumbersOn
\let\oldmaketitle\maketitle
\let\maketitle\relax

\usepackage{sectsty} % change font sizes in sections
\usepackage[font=scriptsize,labelfont=bf]{caption}

\usepackage{amssymb}
\usepackage{amsmath,amsthm}
\usepackage[shortlabels]{enumitem}
\usepackage{hyperref}
\usepackage[dvipsnames]{xcolor}
\usepackage{graphicx}
\usepackage{framed}

\theoremstyle{plain}
\newtheorem{theorem}{Theorem}
\newtheorem{proposition}[theorem]{Proposition}
\newtheorem{lemma}[theorem]{Lemma}
\newtheorem{corollary}[theorem]{Corollary}

\theoremstyle{remark}
\newtheorem{remark}{Remark}

\theoremstyle{definition}
\newtheorem{definition}[theorem]{Definition}

\DeclareMathOperator{\trace}{Tr}
  
\newcommand{\R}{{\mathbb R}}
\renewcommand{\i}{\mathrm{i}}
\renewcommand{\d}{\,\mathrm{d}}
\renewcommand{\r}{\vec{r}}
\renewcommand{\AA}{\vec{A}}
\newcommand{\onehalf}{\tfrac{1}{2}}

\newcommand{\eps}{\varepsilon}
\newcommand{\jpvec}{\vec{j}}

\newcommand\DensSpace {X}
\newcommand\CurSpace {Y}

\newcommand\xreg {(\rho_{\mathrm{reg}},\jpvec_{\mathrm{reg}})}
\newcommand\uKS {(u_\mathrm{KS},\AA_\mathrm{KS})}
\newcommand\uExt {(u_\mathrm{ext},\AA_\mathrm{ext})}
\newcommand\xp {(\rho,\jpvec)}
\newcommand\xpi {(\rho_i,\jpvec_i)}

\newcommand \affilOslo {\affiliation{
    Hylleraas Centre for Quantum Molecular Sciences, University of Oslo, Norway}}
\newcommand \affilHamburg {\affiliation{
    Max Planck Institute for the Structure and Dynamics of Matter, Hamburg, Germany}}

\title{Kohn--Sham theory with paramagnetic currents: compatibility and functional differentiability} %Paramagnetic Current-Density-Functional Theory}

\author{Andre Laestadius}
\email{andre.laestadius@kjemi.uio.no}
\affilOslo

\author{Erik I. Tellgren}
\email{erik.tellgren@kjemi.uio.no}
\affilOslo

\author{Markus Penz}
\affilHamburg

\author{Michael Ruggenthaler}
\affilHamburg

\author{Simen Kvaal}
\affilOslo

\author{Trygve Helgaker}
\affilOslo

\begin{document}
\sectionfont{\fontsize{15pt}{17pt}\selectfont}
\subsectionfont{\fontsize{13pt}{15pt}\selectfont}
\fontsize{10pt}{12pt}\selectfont

\twocolumn[{
\begin{@twocolumnfalse}
\oldmaketitle
\begin{abstract}{\small\noindent
Recent work has established Moreau--Yosida regularization as a mathematical tool to achieve rigorous functional differentiability in density-functional theory. In this article, we extend this tool to paramagnetic current-density-functional theory, the most common density-functional framework for magnetic field effects. 
The extension includes a well-defined Kohn--Sham iteration scheme with a partial convergence result. 
To this end, we rely on a formulation of Moreau--Yosida regularization for reflexive and strictly convex function spaces. 
The optimal $L^p$-characterization of the paramagnetic current density $L^1\cap L^{3/2}$ is derived from the $N$-representability conditions.
A crucial prerequisite for the convex formulation of paramagnetic current-density-functional theory, termed \emph{compatibility} between function spaces for the particle density and the current density, is pointed out and analyzed. Several results about compatible function spaces are given, including their recursive construction.
The regularized, exact functionals are calculated numerically for a Kohn--Sham iteration on a quantum ring, illustrating their performance for different regularization parameters.
}\end{abstract}

\centering
\includegraphics[width=0.6\linewidth]{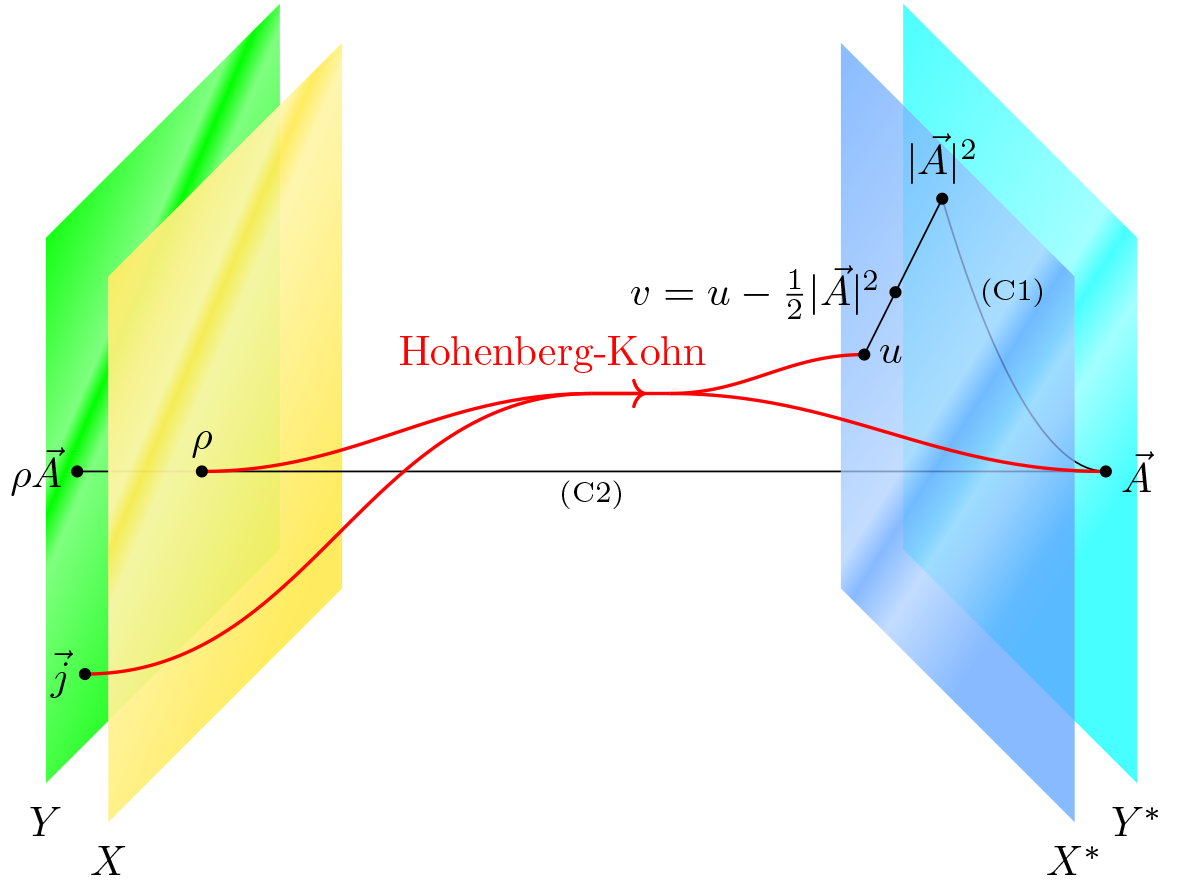}

\vspace{2cm}
\end{@twocolumnfalse}
}]

\newpage
\section{Introduction}
The theoretical foundation of density-functional theory (DFT) was established in a seminal paper by Hohenberg and Kohn \cite{Hohenberg1964}. There it was proven that two potentials that differ by more than a constant cannot share the same ground-state particle density $\rho$ (see Eq.~\eqref{eq:rho-def} for definition). This fact is referred to as the Hohenberg--Kohn (HK) theorem. 
Using this result, the Schr\"odinger equation was replaced by a minimization problem involving a universal density functional (HK variational principle). 
The work by Lieb~\cite{Lieb83} provided an abstract reformulation of DFT that eliminates some technical difficulties with the HK formulation and constitutes a more tractable framework for rigorous mathematical analysis. Lieb's formulation relies on Legendre--Fenchel transformations between the ground-state energy and a universal density functional, analogous to the use of Legendre transformations in thermodynamics and classical mechanics. The HK theorem becomes recast into a fact about subgradients of convex functionals that are mapped one-to-one by Legendre--Fenchel transformations~\cite{Tellgren2012,Kvaal2014}.\newline
\indent As far as practical purposes are concerned, DFT was first converted into a feasible algorithm for electronic structure calculations by Kohn and Sham~\cite{KS}. Here, both the unknown density of the full system and the effective Kohn--Sham (KS) potential for the non-interacting system are solved for in an iterative manner. Even though the important question of convergence of this procedure has been addressed in several works~\cite{Cances2001,CANCES_JCP118_5364,Wagner2013,KSpaper2018,LammertBivariate}, it has only very recently been answered positively for finite-dimensional settings.~\cite{penz2019guaranteed}\newline
\indent The motivation to include current densities and not just the particle density is to obtain a universal functional modelling the internal energy of magnetic systems. In terms of Lieb's Legendre--Fenchel description, the current couples to the vector potential that now also enters the theory to account for the magnetic field. Recent work in current-density-functional theory (CDFT) has been devoted to the extension of the HK theorem, the HK variational principle, and the KS iteration scheme to include current densities~\cite{Vignale1987,Diener,Capelle2002}, as well as to highlight the complexity of such a generalization~\cite{Tellgren2012,LaestadiusBenedicks2014,LaestadiusBenedicks2015}. 
 Other approaches are feasible as well, e.g., the magnetic-field density-functional theory (BDFT) of Grace and Harris~\cite{GRAYCE_PRA50_3089}, where a semi-universal functional is employed instead. There exists also a convexified formulation, in which BDFT and paramagnetic CDFT are related to each other by partial Legendre--Fenchel transformations~\cite{Tellgren2018,REIMANN_JCTC13_4089}. Furthermore, the physically important case of linear vector potentials (uniform magnetic fields) has been theoretically studied in linear-vector-potential density-functional theory (LDFT) without the need to include current densities~\cite{Tellgren2018}. 
Works beyond the current density generalization exist too, e.g., spin-current density-functional theory, reduced-density-matrix-functional theory, Maxwell--Schr\"odinger density-functional theory (MDFT), and quantum-electrodynamical density-functional theory (QEDFT)~\cite{PhysRevB.96.035141,AyersGoldenLevy2006,giesbertz2019one,PhysRevA.97.012504,Rugg2017}. For more generalized density-functional theories see Ref.~\citenum{AyersFuentealba2009} and the references therein, e.g., the kinetic-density-functional theory of Sim \emph{et al.}~\cite{sim2003}, the work of Ayers~\cite{Ayers-JMP-2005} on $k$-density-functional theory, and Higuchi-Higuchi~\cite{HIGUCHI_PRB69_035113} who explored the use of different physical quantities as variables of the theory.\newline
\indent For mathematical reasons, CDFT is formulated in terms of the paramagnetic current density $\jpvec$ (see Eq.~\eqref{eq:jpvec-def} for definition) rather than with the gauge-invariant total current density. A theoretical foundation in the sense of a HK theorem for the total current density has not yet been proven and its existence remains an open question in the general case~\cite{Tellgren2012,LaestadiusBenedicks2014}. However, even if such a result could be shown, a HK variational principle does not exist for the total current density~\cite{LaestadiusBenedicks2015}. Circumventing these problems may require the Maxwell--Schr\"odinger variational principle in place of the standard one~\cite{PhysRevA.97.012504}. For the CDFT that makes use of the paramagnetic current density, it is well-known that there are counterexamples that rule out any analogue of the HK theorem~\cite{Capelle2002,Tellgren2012,LaestadiusBenedicks2014}. Nevertheless, since the particle density and the paramagnetic current density determine the non-degenerate ground state (see Ref.~\citenum{LaestadiusTellgren2018} for results in the degenerate case), a universal Levy--Lieb \cite{Levy79,Lieb83} constrained-search functional can be set up, as done by Vignale and Rasolt~\cite{Vignale1987}. This functional can be extended to a Lieb functional that in this case also depends on the paramagnetic current density (for a first attempt see Ref.~\citenum{Laestadius2014} with the choice of domain $(\rho,\jpvec)\in (L^1\cap L^3)\times \vec{L}^1$). \newline
\indent Since the Lieb functional within standard DFT suffers from non-differentiability~\cite{Lammert2007}, a property that CDFT inherits, we address this particular problem and formulate a regularized theory in a Banach space setting. We here apply our recent work~\cite{KSpaper2018} that also extends the mathematical formalism of paramagnetic CDFT in Ref.~\citenum{Laestadius2014}. The need for differentiability---a fact that is usually overlooked in textbooks---is connected to the variational derivation and analysis of the Kohn--Sham scheme. This task, in the setting of paramagnetic CDFT, is the main aim of this work. \newline
\indent To set up a rigorous CDFT including the corresponding Kohn--Sham scheme, which is borrowed from our previous work~\cite{KSpaper2018} and here baptized Moreau--Yosida--Kohn--Sham optimal damping algorithm (MYKSODA), we introduce and discuss the condition of \emph{compatibility} between function spaces 
for the scalar and vector potentials on the one hand and 
for the paramagnetic current and the total physical current densities on the other. This condition is necessary both for the convex formulation of CDFT and the subsequent Moreau--Yosida regularization procedure. Moreover, to maintain compatibility the regularization procedure requires a 
Banach space formulation and we make use of our results employing \emph{reflexive} Banach spaces~\cite{KSpaper2018}. In this respect the approach presented here differs from that in standard DFT where a Hilbert space formulation has been previously considered~\cite{Kvaal2014}, which does not allow the necessary compatibility in CDFT. However, to apply the Banach space formulation outlined in Ref.~\citenum{KSpaper2018}, a suitable function space for the paramagnetic current density first needs to be identified. The choice $\vec L^1$ from Ref.~\citenum{Laestadius2014} cannot be used for this purpose since it is not reflexive. It is therefore crucial to first prove that the paramagnetic current density is an element of $L^p$ for some $1<p<+\infty$. We prove in Corollary~\ref{lem:fLp} that each component of $\jpvec$ is an element of $L^{3/2}$ under the assumption of finite kinetic energy.\newline
\indent This article is structured as follows. After introducing the basic quantum-mechanical model for paramagnetic CDFT in Sec.~\ref{sec:model}, we define suitable function spaces for particle and current densities in Sec.~\ref{sec:functionSp}. In such a setting, the usual constrained-search functionals of DFT are defined and the energy functional and generalized Lieb functional are subsequently set up in Sec.~\ref{sec:functionals}. These functionals serve as the primary objects for a further study of the theory in terms of convex analysis. Here a first problem arises within CDFT: the lack of concavity of the energy functional. As a further ingredient of a well-defined Kohn--Sham iteration scheme, finiteness of the energy functional (and its concave version) is proven in Sec.~\ref{sec:Efiniteness}. Like the authors recently showed~\cite{KSpaper2018}, the variational Kohn--Sham construction can only be rigorously set up for a regularized theory. The respective form of Moreau--Yosida regularization is introduced and applied to the setting at hand in~Sec.~\ref{sec:MY}. Finally, the stage is set for a discussion of the Kohn--Sham iteration scheme in Sec.~\ref{sec:KS} and its precise formulation as MYKSODA in Sec.~\ref{sec:MYKSODA}. We note possible convergence issues in the particular setting of a two-particle singlet state in Sec.~\ref{sec:KSconvergence}. We conclude in Sec.~\ref{sec:numerApp} with a numerical study of the MYKSODA. For readers that are less familiar with Lebesgue spaces and functional analysis in general we recommend Refs.~\citenum{LiebLoss} and~\citenum{Teschl}, while for the tools borrowed from convex analysis we point to Ref.~\citenum{VANTIEL}.

\section{Paramagnetic CDFT}

\subsection{Ground-state model}
\label{sec:model}

In what follows, we consider the Hamiltonian of an $N$-electron system to be specified by an external scalar potential $v:\mathbb R^3 \to \mathbb R$ and an external vector potential $\AA: \mathbb R^3\to\mathbb R^3$. The components of $\AA$ and other vectors are denoted $A^k$ and are not to be confused with the Euclidean norm squared, $\vert \AA\vert^2= \AA\cdot \AA$. The physical kinetic energy operator and electron-electron repulsion are given by (in SI based atomic units),
\begin{align*}
T(\AA) & = \sum_{k=1}^N \frac{1}{2} \left(-\i\nabla_{\r_k} + \AA(\r_k) \right)^2,
\\
W & = \sum_{k=2}^{N} \sum_{l=1}^{k-1} \frac{1}{\vert \r_{k} - \r_l\vert }.
\end{align*}
The full Hamiltonian then reads
\begin{equation}\label{eq:Hdef}
H^{\lambda}(v,\AA) = T(\AA) + \lambda W + \sum_{k=1}^N v(\r_k),
\end{equation}
where a scale factor $\lambda \geq 0$ is included in front of the electron-electron repulsion term. This means $H^0$ is the Hamiltonian of a non-interacting system while  the usual interacting system is retrieved with $H^1$. This extra parameter is motivated by its usefulness when addressing the KS theory and is standard in the literature. \newline
\indent We consider wavefunctions $\psi = \psi(\vec x_1, \dots,\vec x_N)$, where $\vec x_k = (\r_k,s_k)$ is the spatial and spin coordinate of the $k$-th particle. The wavefunctions are antisymmetric elements of the $N$-electron space $L^2((\mathbb R^3\times \{\uparrow,\downarrow\})^{N})$, i.e., the usual Hilbert space of quantum mechanics. We will be interested in ground-state CDFT, where several options are available to treat the spin degrees of freedom. Firstly, we could formulate a theory for the global ground state, obtained through minimization over all spin degrees of freedom, in which case a spin-Zeeman term could also be included in the Hamiltonian. Secondly, we could instead formulate a theory for the lowest singlet state ($S^2=0$) or the lowest state with some other prescribed value of the spin quantum number $S^2$. For simplicity, we formulate a ground-state CDFT for the lowest singlet energy and adapt our notation accordingly. However, our analysis is mostly independent of this choice and applies equally well to a theory for global ground states. Thus without loss of generality the spin coordinate will in the sequel be omitted. \newline 
\indent All wavefunctions are assumed to have finite kinetic energy,
\begin{equation*}
\mathcal K( \psi ) = \frac{1}{2} \sum_{k=1}^N \int_{\R^{3N}} |\nabla_{\vec r_k} \psi|^2 \d \r_1 \dots \d \r_N  <+\infty.
\end{equation*}
We further assume $L^2$ normalization of $\psi$ and denote the $L^p$ norm by $\Vert \cdot \Vert_p$, $1\leq p \leq \infty$. Henceforth the particle number $N$ will be fixed and we define the set of admissible wavefunctions
\begin{align*}
    \mathcal W_N = \{ \psi: \Vert \psi \Vert_2 &= 1, \mathcal K(\psi) <+\infty, \\
    &\text{ $\psi$ is antisymmetric}\}.
\end{align*}
Moreover, $\gamma_\psi = \vert \psi \rangle \langle \psi \vert$ denotes the density matrix of a pure state $\psi$ and  
$\mathcal P_N = \{  \gamma_\psi : \psi\in \mathcal W_N  \}$ 
is the set of such states. The set of mixed states is given by (where the sum over $n$ can be infinite)
\begin{align*}
\mathcal D_N = &\left\{\gamma = \sum_n p_n \gamma_n  \, : \, \gamma_n\in \mathcal P_N, \right. \\
&\left.\;\;
\sum_n p_n =1, p_n \geq 0, \trace(\gamma T(\vec 0)) < +\infty \right\}.
\end{align*}
Note that in case $\gamma = \gamma_\psi \in \mathcal P_N$ we have 
$\trace(\gamma T(\vec 0)) = \mathcal K(\psi)$.

The energy functional for the ground-state energy can be written in the following alternative forms
\begin{equation}\label{eq:E-functional}
\begin{aligned}
E^{\lambda}(v,\AA) &= \inf_{\gamma\in \mathcal D_N} \trace(\gamma H^{\lambda}(v,\AA)) \\
&= \inf_{\gamma \in \mathcal P_N} \trace(\gamma H^{\lambda}(v,\AA)) \\
& = \inf_{\psi\in\mathcal W_N} \langle \psi, H^\lambda(v,\AA) \psi  \rangle .
\end{aligned}
\end{equation}
Thus if a minimizer $\gamma\in \mathcal D_N$ exists one can always also obtain a pure ground state selected from one of the eigenvectors of $\gamma$.

\subsection{Function spaces for densities}
\label{sec:functionSp}

For any $\psi\in \mathcal W_N$, we define the particle density and the paramagnetic current density, respectively, according to
\begin{align}
%\begin{split}
\rho_\psi (\r_1) &= N \int_{\mathbb R^{3(N-1)}}   |\psi|^2 \d \r_2 \dots \d \r_N, \label{eq:rho-def} \\
\jpvec_\psi(\r_1) &= N\,\text{Im} \int_{\mathbb R^{3(N-1)}}\psi^*\nabla_{\r_1} \psi \, \d \r_2\dots \d \r_N. \label{eq:jpvec-def}
%\end{split}
\end{align}
The aim of this section is to extract as much information as possible about the regularity of $\rho_\psi$ and $\jpvec_\psi$ in terms of $L^p$ spaces from the assumption that $\psi\in \mathcal W_N$. This will define the sets of admissible densities. 

To avoid confusion, a word or two on our notation is appropriate at this point. Since the paramagnetic current density is the main current density of consideration we omit the usual superscript (or subscript) ``p'' for paramagnetic in $\jpvec^\mathrm{p}$. We  write $\underline{\r} = (\r_1,\dots,\r_N) \in\mathbb R^{3N}$, $\nabla = \nabla_{\underline \r}$,
and let $\r$ denote any $\r_i\in\mathbb R^3$ but typically $\r_1$.
Further $\mathcal H^k(\R^n)$ denotes the Sobolev space that includes all functions in $L^2(\R^n)$ with weak derivatives up to $k$-th order in $L^2(\R^n)$. 
($\mathcal H$ should not be confused with the Hamiltonian $H$.)
Finally, $\vec X = X \times X \times X$ is the triple copy of a Banach space $X$, here mostly used for $L^p$ spaces as $\vec L^p$.
\newline
\indent Hoffmann-Ostenhof and Hoffmann-Ostenhof (see Eq.~(3.10) in Ref.~\citenum{HoffmannOstenhof1977}) and Lieb (Theorem~1.1 in Ref.~\citenum{Lieb83}) have shown that the von Weizs\"acker term involving $\rho_\psi$ is bounded by the kinetic energy of $\psi$, i.e.,
\begin{equation}\label{eq:kin-bound}
\frac 1 2 \int_{\mathbb R^3}| \nabla \rho_\psi^{1/2}|^2 \d \r \leq  \mathcal K(\psi),
\end{equation}
and therefore $\rho_\psi\in \mathcal I_N$ if $\psi\in \mathcal W_N$, where
\[
\mathcal I_N = \left\{ \rho\in L^1(\mathbb R^3): \rho\geq 0, \Vert\rho\Vert_1 = N, \rho^{1/2}\in \mathcal H^1(\mathbb R^3) \right\}
\]
denotes the set of $N$-representable particle densities. Even though the Hilbert space $\mathcal H^2(\R^{3N})$ is the natural domain of the kinetic energy operator, $\psi\in \mathcal H^1(\R^{3N})$ is sufficient to guarantee finite $\mathcal K(\psi)$.
The Sobolev inequality in $\mathbb R^3$ (see, e.g., Theorem~8.3 in Ref.~\citenum{LiebLoss}),
\begin{align}
\left(\int_{\mathbb R^3}|f(\r)|^6 \d \r\right)^{1/3} \leq S \int_{\mathbb R^3}|\nabla f(\r)|^2 \d \r,
\label{Sob}
\end{align}
applied to $f= \rho_\psi^{1/2}$ yields
\begin{equation}\label{Sob2}
\Vert \rho_\psi \Vert_3 \leq S \int_{\mathbb R^3} \vert \nabla \rho_\psi^{1/2}\vert^2 \d \r, \quad S = \frac{4}{3(16\pi^2)^{1/3}}. 
\end{equation}
Consequently $\mathcal I_N \subset L^1\cap L^3$.
\begin{remark}\label{remark:interpolation}
	We make a brief comment concerning interpolation. For $1\leq q \leq p\leq r\leq +\infty$, if $f\in L^q\cap L^r$ then by H\"older's inequality
	\[
	\Vert f \Vert_p^p \leq \Vert f \Vert_r^{r(p-q)/(r-q)}\Vert f \Vert_q^{q(r-p)/(r-q)}
	\]
	and thus $f\in L^p$.
\end{remark}

We proceed by summarizing criteria that follow from the works of Lieb and Kato for the space of particle densities in terms of $L^p$ spaces.

\begin{proposition}[Lieb~\cite{Lieb83} and Kato~\cite{Kato1951}]\label{prop:rhoLp}
For $\psi\in \mathcal W_N$, $\rho_\psi \in L^p$ with $p\in [1,3]$ and in particular $\rho_\psi$ is an element of the Hilbert space $L^2$.  
Moreover, $\psi\in \mathcal H^2$ implies $\rho_\psi \in L^p$ for all $p\in [1,\infty]$.
\end{proposition}

\begin{proof}
The first part follows from Lieb~\cite{Lieb83} and Remark~\ref{remark:interpolation}. By Lemma~3 in Ref.~\citenum{Kato1951}, the assumption $\psi\in \mathcal H^2$ gives $\rho_\psi\in L^\infty$. (Note that $A_i f(\r_i)$ in Ref.~\citenum{Kato1951} corresponds to $(\rho(\r_i)/N)^{1/2}$ in our notation.)
\end{proof}

An obvious limitation for the current density is that every component of $\jpvec_\psi$ is in $L^1$~\cite{Laestadius2014}. Here, by better exploiting the properties of $\psi\in \mathcal W_N$, we will be able to further characterize the set of current densities. \newline
\indent We start by giving some definitions. The kinetic-energy density $\tau_\psi:\mathbb R^3\rightarrow \mathbb R_+$ is given by
\[
\tau_\psi(\r) = N \int_{\mathbb R^{3(N-1)}} 
|\nabla_{\r} \psi |^2 \d \r_2\dots \d \r_N,
\]
and relates to the already defined kinetic energy by $\mathcal K(\psi) = \frac{1}{2} \Vert \nabla \psi\Vert_{2}^2 = \frac{1}{2}\|\tau_\psi\|_1$. We see that $\psi$ being an element of $\mathcal W_N$ guarantees finite $\mathcal H^1$ norm and thus $\mathcal W_N \subset \mathcal H^1$.
Furthermore, let $\r\in \mathbb R^3$ be written $\r=(r^1,r^2,r^3)$ and define the component-wise kinetic energy density
\[
\tau_{\psi}^k(\r) = N \int_{\mathbb R^{3(N-1)}}  |\nabla_{r^k} \psi |^2 \d \r_2\dots \d \r_N
\]
and $\mathcal K^k(\psi) = \frac{1}{2}\|\tau_\psi^k\|_1$.
A direct computation gives that the usual Sobolev norm satisfies
\begin{equation*}
\Vert \psi\Vert_{\mathcal H^1}^2 = \int_{\mathbb R^3}\left( \rho_\psi(\r)/N  +  \tau_\psi(\r)\right) \d \r = 1 + 2\mathcal K (\psi).
\end{equation*}

An important bound that will be used subsequently is the following one
\begin{align}
\left\vert  N \int_{\mathbb R^{3(N-1)}} \psi^*\nabla_{\r}\, \psi \d \r_2\dots \d \r_N \right\vert
&\leq \tau_\psi(\r)^{1/2} \rho_\psi(\r)^{1/2}, \label{eq:taurho}
\end{align}
which is a consequence of the Cauchy--Schwarz inequality. Since both $\tau_\psi^{1/2}$ and $\rho_\psi^{1/2}$ are $L^2$ functions, integration and the Cauchy--Schwarz inequality give $\jpvec_\psi\in \vec L^1$. Indeed
\[
\int_{\mathbb R^3} |\jpvec_\psi(\r)|\d \r \leq \frac 1 2 (2N\,\mathcal K(\psi))^{1/2},
\]
see Proposition~3 in Ref.~\citenum{Laestadius2014}. The idea is now to further extend such $\vec L^p$-characterizations. With $\jpvec_\psi = (j_\psi^1,j_\psi^2,j_\psi^3)$ we now state and prove 

\begin{lemma} \label{lem:fLp}
	Set $f_{\psi}^k=\partial_k \rho_\psi + 2\i j_\psi^k $ for $k=1,2,3$ and $\psi\in \mathcal W_N$. Then $f_{\psi}^k \in L^p(\mathbb R^3)$ for $1\leq p\leq 3/2$.
\end{lemma}

\begin{proof}
Since
\[
f_{\psi}^k(\r) = N \int_{\mathbb R^{3(N-1)}} \psi^*\partial_{r^k} \psi \d \r_2\dots \d \r_N,
\]
we have similar to Eq.~\eqref{eq:taurho} for $1\leq p\leq 3/2$ 
\begin{align}\label{eq:frhotau}
|f_\psi^k(\r)|^p \leq \left( \rho_\psi(\r)  \tau_\psi^k(\r)  \right)^{p/2}    
\end{align}
by the Cauchy--Schwarz inequality.
Next, we use H\"older's inequality with $q$ defined by $p/2 + 1/q = 1$ such that
\begin{align}
	&\left(\int_{\mathbb R^3}|f_{\psi}^k(\r)|^{p}\d \r\right)^{1/p} \nonumber\\
	&\quad\leq \left[ \left(\int_{\mathbb R^3}\rho_\psi(\r)^{pq/2} \d \r\right)^{1/q} \left(\int_{\mathbb R^3}\tau_{\psi}^k(\r) \d \r\right)^{p/2}\right]^{1/p}\nonumber \\
	&\quad =  \Vert \rho_\psi\Vert_{pq/2}^{1/2} (2 \mathcal K^k(\psi))^{1/2}.
\label{eq:PbyH}
\end{align}
To conclude, we note that, by the assumption on $p$, we have $1 \leq pq/2 \leq 3$ and recall that $\rho$ is in $L^1 \cap L^3$.
\end{proof}

Note that $\vec f_\psi$ in Lemma~\ref{lem:fLp} can be seen as a generalized
complex current, similar to the one considered by Tokatly~\cite{Tokatly2011} in a lattice version of time-dependent CDFT. It has also been considered before as ``momentum density''~\cite{TellgrenNrep}.

As a direct consequence of Lemma~\ref{lem:fLp}, we have our main result about function spaces for the current density
\begin{theorem}\label{th:jLp}
	For $\psi\in\mathcal W_N$, each component $j_\psi^k$ of the paramagnetic current density $\jpvec_\psi$ is in $L^p$ for any $1\leq p\leq 3/2$ and we write $\jpvec_\psi \in \vec L^p$. In particular, we have 
\begin{align*}
\Vert j_\psi^k \Vert_1 \leq(2N\, \mathcal K^k(\psi))^{1/2}, \quad \Vert j_\psi^k\Vert_{3/2} \leq  S^{1/2} 2\mathcal K^k(\psi),
\end{align*}
with the constant $S$ given by Eq.~\eqref{Sob2}.
\end{theorem}

\begin{proof}
Set $p=1$ and apply Lemma~\ref{lem:fLp}. Since $pq/2=1$ and $\Vert \rho_\psi\Vert_1=N$, Eq.~\eqref{eq:PbyH} gives 
	\begin{align*}
	\Vert j_\psi^k \Vert_1 \leq \int_{\mathbb R^3}|f_{\psi}^k(\r)|\d \r \leq (2N\, \mathcal K^k(\psi))^{1/2}.
	\end{align*}
With the choice $p=3/2$ instead, we have $pq/2= 3$. Eq.~\eqref{eq:PbyH} then reduces to
	\begin{align*}
	\Vert j_\psi^k\Vert_{3/2} \leq S^{1/2} 2 \mathcal K^k(\psi),
	\end{align*}
where we have also used the Sobolev inequality \eqref{Sob} and Eq.~\eqref{eq:kin-bound}.
Using interpolation (see Remark~\ref{remark:interpolation}), it follows
	\[
	\Vert j_\psi^k\Vert_p\leq \left[ N^{3-2p} S^{3p-3} (2 \mathcal K^k(\psi))^{4p-3}\right]^{\frac 1 {2p}}
	\]
and thus $j_\psi^k\in L^p$ for $1\leq p \leq 3/2$.
\end{proof}

From the proof of Lemma~\ref{lem:fLp} (the Cauchy--Schwarz inequality applied to Eq.~\eqref{eq:frhotau} with $p=2$), we obtain the current correction to the von Weizs\"acker kinetic energy (see Ref.~\citenum{BATES_JCP137_164105} and note that this sharpens Theorem~14 in Ref.~\citenum{Laestadius2014})
\begin{align}\label{eq:Furches}
    \frac 1 2 \int_{\R^3} \vert \nabla\sqrt{\rho_\psi} \vert^2\d\r + \frac 1 2\int_{\R^3} \frac{\vert \jpvec_\psi \vert^2}{\rho_\psi} \d\r\leq \mathcal K(\psi).
\end{align}
The well-known inequality $\frac 1 2\int_{\R^3} \vert \jpvec_\psi \vert^2 \rho_\psi^{-1} \d\r\leq \mathcal K(\psi)$ is immediate from Eq.~\eqref{eq:Furches}. \newline
\indent 
We next note that the space for the current density cannot be further restricted since $j_\psi^k\notin L^p$, $p>3/2$, for some $\psi\in\mathcal W_N$. Before proving this fact, a further characterization of $\jpvec_\psi$ using $\psi \in \mathcal H^2$ is given.

\begin{proposition}\label{prop:L2}
$\psi \in \mathcal H^2$ implies $\jpvec_\psi\in \vec L^2$.
\end{proposition}

\begin{proof}
	Suppose  $\psi \in\mathcal H^2$. By Proposition \ref{prop:rhoLp} it follows that $\rho_\psi\in L^\infty$. Then from Eq.~\eqref{eq:taurho}
	\begin{align*}
	\int_{\mathbb R^3} &\left\vert N \int_{\mathbb R^{3(N-1)}} \psi^*\nabla_{\r}\,\psi \d \r_2\dots \d \r_N \right\vert^2 \d \r \nonumber \\
	& \leq \int_{\mathbb R^3} \tau_\psi(\r) \rho_\psi(\r) \d \r \nonumber \\
	& \leq 	\Vert \rho_\psi \Vert_{\infty} \, 2\mathcal K(\psi)  <+\infty,
	\end{align*}
	which gives $\jpvec_\psi  \in \vec L^2$. (Note that we could have argued by means of the Cauchy--Schwarz inequality to obtain $\int_{\mathbb R^3}  \vert \jpvec \vert^2 \d \r \leq \Vert\rho\Vert_\infty \int_{\mathbb R^3} \vert \jpvec \vert^2\rho^{-1} \d \r$, where the integral over $\vert \jpvec \vert^2\rho^{-1}$ is bounded in terms of $\mathcal K(\psi)$.)
\end{proof}

\begin{proposition}\label{prop:jnotLp}
For $N=1$, there exists $\phi \in\mathcal W_N$ such that $\jpvec_\phi\notin \vec{L}^p$, for every $p>3/2$.
\end{proposition}

\begin{proof}
	Consider the function 
	\[
	\phi(\r) = r^{\alpha/2} \mathrm{e}^{\i \xi(r)},\quad \xi(r) = r^\beta,\quad r= \vert \r \vert,
	\]
    where $0\leq r \leq 1$, $\alpha > -1$, and $2\beta >-\alpha -1$. It then holds that $\phi \in \mathcal H^1(\mathbb R^3)$, since
	\begin{align*}
	\int_0^1 |\phi|^2 r^2 \d r &= \int_0^1 r^{\alpha+2} \d r =\left[\frac{r^{\alpha + 3}}{\alpha + 3}\right]_0^1 = \frac 1 {\alpha +3},\\ 
	\int_0^1 |\nabla\phi|^2 r^2 \d r &= \int_0^1 |\nabla (r^{\alpha/2})|^2 r^2 \d r + \int_0^1 |\nabla \xi|^2 r^{\alpha + 2} \d r  \\
	& = \frac \alpha 4 ^2 \int_0^1 r^{\alpha} \d r  + \beta^2 \int_0^1 r^{\alpha + 2\beta} \d r \\
	&=\frac{\alpha^2}{4(\alpha+1)} + \frac{\beta^2}{\alpha + 2\beta+ 1}. 
	\end{align*}
    We wish to show that $\jpvec_\phi \notin  \vec L^p$, for all $p>3/2$ for some choice of $(\alpha,\beta)$ in the set 
	\[\Theta =\{ \alpha>-1,\beta>-(\alpha + 1)/2  \}.
	\]
	Note that $\jpvec_\phi = \text{Im}\,\phi^*\nabla\phi = r^\alpha \nabla r^\beta$. Let $\delta$ be an arbitrarily small positive number and set $\alpha= -1+3\delta$. Then with $\beta = -\delta>-3\delta/2$ we have $(\alpha,\beta)\in\Theta$. 
	For all $p> 3/(2(1-\delta))$ we obtain
	\begin{align*}
	\int_0^\infty \vert \jpvec_\phi \vert^p r^2 \d r &\geq \delta^p \int_0^1 r^{2p(\delta-1)+2 } \d r \\
	&= \delta^p \lim_{\eps\rightarrow 0 + } \left[ \frac{r^{2p(\delta-1)+ 3}}{2p(\delta-1)+ 3}\right]_\eps^1 =+\infty.
	\end{align*}
Since any $\delta>0$ is allowed, we conclude $\jpvec_\phi \notin \vec L^p$, $p>3/2$.
\end{proof}

\begin{remark}\label{remark-rho-Lbigger3}
Note that the same counterexample also shows that $\rho_\phi \notin L^p$, for $p>3$.
\end{remark}

From $\psi \in \mathcal W_N$, with Proposition~\ref{prop:rhoLp} and Remark~\ref{remark-rho-Lbigger3} we have thus arrived at the well-known optimal choice $L^1\cap L^3$ of $L^p$ spaces for the particle density~\cite{Lieb83}. Similarly, by Theorem~\ref{th:jLp} and Proposition~\ref{prop:jnotLp} the optimal choice for the paramagnetic current density is $\vec L^1 \cap \vec L^{3/2}$. Note that densities and currents that are not from these spaces cannot be represented by admissible wavefunctions $\psi \in \mathcal W_N$. Later this choice will be further limited by the demands coming from \emph{compatibility} (see Sec.~\ref{sec:compatibility}), reflexivity, and strict convexity (in connection with the regularized KS iteration scheme, Sec.~\ref{sec:MYKSODA}). \newline
\indent We summarize this section with some definitions and a concluding corollary. We also refer to Refs.~\citenum{ENGLISCH1983,Lieb83,LiebSchrader,TellgrenNrep,Laestadius2014b,WagnerFollowUp} for further discussions on this topic (not only confined to CDFT). 
\begin{definition}
    We say that a density pair $(\rho,\jpvec)$ is \emph{$N$-representable} if there exists a $\psi\in \mathcal W_N$ such that $\rho_\psi = \rho$ and $\jpvec_\psi= \jpvec$. If such a $\psi$ is the ground state of some $H^\lambda(v,\AA)$, then $(\rho,\jpvec)$ is ($\lambda$) \emph{$v$-representable}. Furthermore, we distinguish between \emph{fully interacting} ($\lambda=1$) and \emph{non-interacting} ($\lambda=0$) $v$-representability. If $\psi\in \mathcal W_N$ is replaced by $\gamma\in \mathcal D_N$ we call the above property \emph{ensemble} $v$-\emph{representability}.
\end{definition}

\begin{corollary}\label{cor:Nrepdensities}
The set of $N$-representable density pairs $(\rho,\jpvec)$ is a subset of
$(L^1\cap L^3) \times (\vec L^1 \cap \vec L^{3/2})$.
\end{corollary}

The above set is given in terms of $L^p$ spaces only. There are other well established constraints such as $\rho(\r)\geq 0$, $\int_{\R^3} \rho \d\r = N$, and $\int_{\R^3} \vert \jpvec \vert^2 \rho^{-1}\d\r < +\infty$~\cite{Lieb83,LiebSchrader,Laestadius2014}.

It will later be important to impose a reflexive (R) and strictly convex Banach space setting, and we therefore define such a space that includes the set of $N$-representable density pairs. Hanner's inequality~\cite{hanner1956uniform} shows that $L^p$, $1<p<\infty$, is even a uniformly convex space which implies both strict convexity and reflexivity.
\begin{definition} \label{def:XRYR}
     %For $1< s \leq 3$, let
      %  $$X_s \times Y_s = (L^{s} \cap L^3) \times (L^{\frac {2s}{s+1}} \cap  \vec L^{3/2}).$$
%In particular, we set 
$X_\mathrm{R} \times Y_\mathrm{R} = L^3 \times  \vec L^{3/2}$.
\end{definition}

\subsection{Constrained-search functionals}
\label{sec:functionals}

To formulate a rigorous CDFT, several requirements need to be placed on densities and potentials. Some of these requirements are related to $N$-representability and thus do not amount to any restriction, but merely exclude irrelevant densities that are invalid in the sense that they cannot arise from any quantum-mechanical state $\psi\in\mathcal W_N$. Other requirements amount to assumptions about the ground-state densities or restrictions on the external potentials that can be considered.

The universal part of the Hamiltonian $H^\lambda(v,\AA)$ in Eq.~\eqref{eq:Hdef}, independent of the potential pair $(v,\AA)$, is $H^\lambda(0,\vec 0) = T(\vec 0) + \lambda W$. On a Banach space $\DensSpace \times \CurSpace \subset \DensSpace_\mathrm{R} \times \CurSpace_\mathrm{R}$ of measurable functions
for particle densities ($X$) and current densities ($Y$), define the universal Levy--Lieb-type functionals $F_\mathrm{VR}^\lambda$ and $F_\mathrm{VR,DM}^\lambda$ according to
	\begin{align*}
	F_\mathrm{VR}^\lambda(\rho,\jpvec) &= \inf_{\psi\in \mathcal W_N} \left\{\langle \psi, H^\lambda(0,\vec 0) \psi\rangle : \psi \mapsto (\rho,\jpvec) \right\}, \\
	F_\mathrm{VR,DM}^\lambda(\rho,\jpvec) &= \inf_{\gamma\in \mathcal D_N} \left\{\trace (\gamma H^\lambda(0,\vec 0)) : \gamma \mapsto (\rho,\jpvec) \right\}.
	\end{align*}
These functionals are derived from the parts of the energy expressions in Eq.~\eqref{eq:E-functional} that are independent of the potential pair $(v,\AA)$. In analogy with Eq.~\eqref{eq:E-functional} we generally have two possibilities, searching either over pure or mixed states. 
If a given density pair $(\rho,\jpvec)$ cannot be represented by a pure or mixed state, then the value of the functional will just be set to $+\infty$ by definition. 
Unlike the ground-state energy functional in Eq.~\eqref{eq:E-functional}, the pure and mixed search domains do not yield equivalent results, since the former is subject to more severe representability restrictions~\cite{TellgrenNrep}. 
	
\begin{remark} The density functionals are here denoted ``VR'' which stands for Vignale and Rasolt to credit their work~\cite{Vignale1987}. We remark that we could just as well have chosen to credit Levy, Valone, and Lieb due to the obvious counterpart in DFT~\cite{Levy79,Lieb83,Valone80}. ``DM'' refers to the use of density matrices for mixed states and was introduced for the DFT constrained-search functional by Valone in Ref.~\citenum{Valone80}.
\end{remark}
\begin{remark}
At this point we wish to keep the setting general and thus the space $X\times Y$ is not specified any closer. This setting includes the choice $X = L^1\cap L^3$, $Y= L^1\cap L^{3/2}$ such that all information from the previous section on function spaces is used. However, if a regularized theory is to be obtained, we need to have a reflexive Banach space setting and therefore we also consider the less restrictive choice $X =X_\mathrm{R}$, $Y= Y_\mathrm{R}$. 
\end{remark}

The Levy--Lieb-type functional $F_\mathrm{VR}(\rho,\jpvec)$ is not convex, see Proposition~8 in Ref.~\citenum{Laestadius2014}. 
Yet by the linearity of the map $\gamma \mapsto (\rho,\jpvec)$ it follows that $F_\mathrm{VR,DM}(\rho,\jpvec)$ \emph{is} convex.
Both functionals are \emph{admissible}~\cite{KH2015} in the sense that they can be used to compute the ground-state energy. 

Since the energy expression will naturally include integrals over couplings of potentials with densities, it is helpful to introduce the notion of dual pairings (between elements of dual Banach spaces). 
For measurable functions $f,g$ with domain $\R^3$ let
\[
\langle f,g\rangle = \int_{\mathbb R^3} f(\r) g(\r) \d \r,
\]
whenever the integral is well-defined in $\mathbb{R}\cup\{\pm\infty\}$,
and similarly for vector-valued functions $\vec f,\vec g$, but with the pointwise product replaced by $\vec f \cdot \vec g$.
Then the energy expressions in Eq.~\eqref{eq:E-functional} can be written as
\begin{align*}
E^\lambda(v,\AA) = \inf_{ \substack{ (\rho,\jpvec)  \in X\times Y } }
\Big\{  F_\mathrm{VR}^\lambda&(\rho,\jpvec) + \langle \AA, \jpvec \rangle \\ 
& + \langle v+\onehalf \vert \AA \vert^2, \rho \rangle \Big\},
\end{align*}
and equivalently by employing $F^\lambda_\mathrm{VR,DM}$ defined with mixed states instead of $F_\mathrm{VR}^\lambda$. 
In particular, $\lambda=1$ corresponds to the fully interacting system, and $\lambda=0$ to a non-interacting one. \newline
\indent At the outset, the formulation of paramagnetic 
CDFT relies on a decomposition of the total
kinetic energy into canonical kinetic energy, the paramagnetic term
$ \langle \AA,\jpvec \rangle$,
and the diamagnetic term $\langle  |\AA|^2, \rho \rangle$, with each of the
terms separately finite~\cite{Tellgren2012,LaestadiusBenedicks2014,SK-TH-preprint}. As in standard
DFT, the electrostatic interaction with the
external potential, $\langle v,\rho\rangle$, needs to be finite
too. In the convexified form, a new potential variable is formed by
absorbing the diamagnetic term into the scalar potential $u = v +
\frac 1 2 |\AA|^2$~\cite{Tellgren2012,Laestadius2014}. Minimally, then, the underlying function spaces should be such
that
\begin{align*}
|\langle u,\rho \rangle| < +\infty \quad \text{and} \quad |\langle \AA,\jpvec \rangle| < +\infty.
\end{align*}
A convex formulation achieves this automatically as it requires the
stronger condition that densities and potentials are elements of dual Banach spaces:

\begin{definition}[Density-potential duality]\label{def:dual}
	We say that there is \emph{density-potential duality}, or just
	\emph{duality}, when densities and potentials are confined to dual
	Banach spaces
	\begin{align*}
	\text{(D1)} & \quad \rho \in \DensSpace \quad \text{and} \quad u \in \DensSpace^*, \\
	\text{(D2)} & \quad \jpvec \in \CurSpace \quad \text{and} \quad \AA \in \CurSpace^*.
	\end{align*}
\end{definition}

\begin{remark}
At this moment we do not assume any more specific properties for $X\subset X_\mathrm{R}$ and $Y \subset Y_\mathrm{R}$ besides duality. However, reflexivity and strict convexity of $\DensSpace$ and $\CurSpace$ are
additional assumptions that will become important in Sec.~\ref{sec:IIIks}.
\end{remark}
\begin{remark} Note that $u\in X^*$, $X= L^1\cap L^3$ and $\AA\in Y^*$, $Y=L^1\cap L^{3/2}$ imply $u \in L^{3/2} + L^\infty$ and $\AA \in \vec L^3 + \vec L^\infty$. For the condition on the original scalar potential $v$ see the next section on \emph{compatibility}. As far as the vector potential is concerned, the restrictions on $\AA$ are stronger than the familiar setting of $\AA \in L_\mathrm{loc}^2$ (see e.g., Ref.~\citenum{LiebLoss}), which is implied by $\AA \in \vec L^3 + L^\infty$. Also the reflexive setting with $X= X_\mathrm{R}$ and $Y= Y_\mathrm{R}$, where $u\in L^{3/2}$ and $\AA\in \vec L^3$, implies $\AA \in L_\mathrm{loc}^2$ again. Moreover, $u\in L_\mathrm{loc}^{3/2}$ is a natural assumption~\cite{LiebLoss}. We remark that our consideration of dual spaces is mathematically motivated and not a physical necessity. A truncated space domain is in many cases needed to cover the usual potentials of physical systems (see Ref.~\citenum{Kvaal2014} for a discussion on this topic).
\end{remark}

Since the potentials $(v,\AA)$ are \emph{not} paired linearly with the densities $(\rho,\jpvec)$, the functional $E^\lambda$ defined in this way is \emph{not} concave. The change of variables $u = v + \tfrac{1}{2} \vert\AA \vert^2$ results in a convexification of paramagnetic CDFT, meaning that
\begin{equation}\label{eq:E-transform}
\bar{E}^{\lambda}(u,\AA) = E^\lambda(u-\onehalf\vert \AA \vert^2, \AA)
\end{equation}
is a \emph{jointly concave} functional~\cite{Tellgren2012}. The consequences of this variable change for the choice of function spaces will be discussed in Sec.~\ref{sec:compatibility}. The price to pay for concavity is a convoluted gauge symmetry. For all scalar fields $\chi$ with gradients in the same function space as $\AA$ one has
\begin{align*}
\bar{E}^{\lambda}(u,\AA) &= E^{\lambda}(u-\onehalf \vert \AA \vert^2,\AA) \\
&= E^{\lambda}(u - \tfrac{1}{2} \vert \AA \vert^2,\AA+\nabla\chi)  \\
&= \bar{E}^{\lambda}(u+\AA\cdot\nabla\chi+\onehalf |\nabla\chi|^2,\AA+\nabla\chi).
\end{align*}
But the benefit is much greater, making $\bar{E}^{\lambda}$ jointly concave in both potentials and highlighting the linear structure of coupling between potentials and densities 
\begin{align*}
\bar E^\lambda(u,\AA) &= \inf_{  (\rho,\jpvec) \in X\times Y  } 
\left\{  F_\mathrm{VR}^\lambda(\rho,\jpvec) +  \langle \AA, \jpvec \rangle + \langle u, \rho \rangle \right\}.
\end{align*}

The convex formulation of paramagnetic CDFT can be outlined as follows. Let the dual space of $X\times Y$ be given by $ \DensSpace^* \times \CurSpace^*$. 
We define the generalized Lieb functional $F^\lambda(\rho,\jpvec)$ as the supremum over the energy plus the linear coupling between densities and potentials, i.e., 
\begin{equation}\label{eq:F-LF}
F^\lambda(\rho,\jpvec) = \sup_{ (u,\AA)\in X^*\times Y^* } \left\{\bar E^\lambda(u,\AA)  - \langle \AA,\jpvec \rangle  - \langle u,\rho\rangle \right\}.
\end{equation}
Such a functional is by construction convex. To extract more from Eq.~\eqref{eq:F-LF} we first need

\begin{definition}\label{def:defs}
Let $B$ be a Banach space with dual $B^*$, $f:B\to \R \cup \{\pm \infty\}$, and $g: B^*\to \R \cup \{\pm \infty\}$.

\begin{itemize}
\item[(i)] If $f$ is convex, lower semi-continuous, has $f > -\infty$, and is not identically equal to $+\infty$, then it is called \emph{closed convex} and we write $f\in \Gamma_0(B)$. Analogously, with weak-* lower semi-continuity we define $\Gamma_0^*(B^*)$. We also introduce the sets $\Gamma(B) = \Gamma_0(B) \cup \{\pm \infty\}$ and $\Gamma^*(B^*) = \Gamma^*_0(B^*) \cup \{\pm \infty\}$.

\item[(ii)] Following Refs.~\citenum{Lieb83,Kvaal2014}, we define the (skew) conjugate functionals (Legendre--Fenchel transformations),
\begin{align*}
    f^\wedge(b^*) &= \inf_{b\in B} \{ f(b) + \langle b^*,b\rangle  \} \in -\Gamma^*(B^*), \\
    g^\vee(b) &= \sup_{b^*\in B^*} \{ g(b^*) - \langle b^*,b\rangle  \} \in \Gamma(B).
\end{align*}
\end{itemize}

\end{definition}

\begin{definition}
\label{def:norms-union-plus}
The standard norms for the intersection of two Banach spaces $B,B'$ and their set sum are given by the following expressions~\cite{liu1969sums}
\begin{align*}
\|b\|_{B \cap B'} &= \max\{\|b\|_B, \|b\|_{B'}\}, \\
\|b\|_{B + B'} &= \inf\{\|b\|_B + \|b'\|_{B'} : b \in B, b' \in B', \\
& \qquad \qquad  b = b + b' \}.
\end{align*}
\end{definition}

Theorem~3.6 in Lieb~\cite{Lieb83} can be straightforwardly generalized to the statement that $F^\lambda$ is lower semi-continuous on the space $X\times Y = (L^1\cap L^3)\times (\vec L^1 \cap \vec L^{3/2})$, with topology defined by the norm from the above definition. (See also Proposition~12 in Ref.~\citenum{Laestadius2014} where this was done for $(L^1\cap L^3)\times \vec L^1$.)
%In particular, we have 
Also, by the same argument, we have for the reflexive setting

\begin{lemma}\label{lemma:Gamma}
$F^\lambda \in \Gamma_0(X_\mathrm{R}\times  Y_\mathrm{R})$.
\end{lemma}

Since $F^\lambda$ is convex \emph{and} lower semi-continuous, this clears the way for the application of powerful tools from convex analysis. Because Eq.~\eqref{eq:F-LF} is already the Legendre--Fenchel transformation of the energy functional $\bar E^\lambda$ we are able to switch back from $F^\lambda$ to $\bar E^\lambda$ with the inverse transformation (see Theorem~1 in Ref.~\citenum{KSpaper2018})
\begin{equation}\label{eq:E-LF}
\bar{E}^{\lambda}(u,\AA) = \inf_{ (\rho,\jpvec) \in X\times Y  } \left\{ F^{\lambda}(\rho,\jpvec) + \langle \AA, \jpvec \rangle + \langle u,\rho\rangle \right\}.
\end{equation}
Using the notation from Definition~\ref{def:defs}~(ii), we sum up the situation as
\begin{align*}
    \bar E^\lambda = \left(F_\mathrm{VR}^\lambda\right)^\wedge = \left(F_\mathrm{VR,DM}^\lambda\right)^\wedge = \left(F^\lambda \right)^\wedge, 
    \quad F^\lambda = \left( \bar E^\lambda\right)^\vee. 
\end{align*}
Moreover the closed convex $F^\lambda$ is the smallest possible admissible functional, 
\[ F^\lambda \leq F^\lambda_{\text{VR,DM}} \leq F^\lambda_{\text{VR}}. \]
Solving the variational problem in Eq.~\eqref{eq:E-LF} is the general task of CDFT.

\begin{remark}
It is to the best of our knowledge an open question whether $F_\mathrm{VR,DM}^\lambda$ is lower semi-continuous and hence equal to $F^\lambda$ in the context of paramagnetic CDFT.
\end{remark}

\subsection{Compatibility of function spaces}
\label{sec:compatibility}

Duality as in Definition~\ref{def:dual} is not strong enough to guarantee finiteness of the diamagnetic
term. It also does not guarantee another
natural condition on the function space for the diamagnetic contribution to the current density that we call \emph{compatibility}.

\begin{definition}
	\label{defCompatibility}
	The density function space $\DensSpace$ and the current density function space
	$\CurSpace$ are said to be \emph{compatible} if, for all $\rho \in \DensSpace$ and all
	$\AA \in \CurSpace^*$,
	\begin{equation*}
	\text{(C1)} \quad |\AA|^2 \in \DensSpace^* \quad \text{and} \quad \text{(C2)} \quad \rho \AA \in \CurSpace.
	\end{equation*}
\end{definition}

We emphasize both conditions in the definition above, as they have different physical interpretations, although it will be seen in Theorem \ref{thmCompatilityPartEquiv} that (C1) and (C2) are equivalent.
The first condition (C1) requires that the
scalar potential $v$ and $\vert \AA\vert^2$ share the same function space, so that
changes of variables between $v = u -\frac 1 2 |\AA|^2$ and $u = v + \frac 1 2 |\AA|^2$ stay
within the space $\DensSpace^*$. The second condition (C2) requires that the paramagnetic
contribution, $\jpvec$, and the diamagnetic contribution, $\rho \AA$,
to the total current density share the same function space $\CurSpace$. \newline
\indent Compatibility also ensures a sensible behavior under gauge transformations. Duality imposes a restriction on the gauge transformations that are allowed within the theory. Any gauge function $\chi$ that is used to transform $\AA$ to $\AA' = \AA + \nabla\chi$ must satisfy $\nabla\chi \in \CurSpace^*$. If $(v,\AA)$ has the ground-state density pair $(\rho,\jpvec)$, the gauge transformed potential pair $(v,\AA')$ would be expected to have the ground-state density pair $(\rho,\jpvec' = \jpvec + \rho\nabla\chi)$. The second compatibility condition (C2) ensures that $\jpvec' \in \CurSpace$, so that ground-state density pairs are never lost after an allowed gauge transformation.

The following theorem shows that the two compatibility
conditions are in fact equivalent. First, we mention a fundamental result that will be used in the sequel. Suppose we are given a general Banach space $B$ of measurable functions and a measurable function $g$. We can then check that $g$ is contained in the dual $B^*$ by verifying that the pairing $\langle g,f\rangle = \int\! g \, f \d\r$ is finite for all $f \in B$, see Appendix~\ref{app:functional} for a full proof of this statement.

\begin{theorem} \label{thmCompatilityPartEquiv}
    If $X,Y$ and their duals are Banach spaces of measurable functions, then (C1) $\iff$ (C2).
\end{theorem}

\begin{proof}
	Part 1 ($\impliedby$): From (C2), we have that $$\int_{\R^3} |\rho \AA\cdot\AA'| \d\r <+\infty$$ for all $\rho \in \DensSpace$ and all
	$\AA, \AA' \in \CurSpace^*$. Specialization to the case $\AA' = \AA$
	yields (C1). \\
	\indent Part 2 ($\implies$): Suppose (C2) is false, i.e., $\rho \AA \notin
	\CurSpace$. Then there exists an $\AA' \in \CurSpace^*$ such that
	\begin{equation*}
	 \int_{\R^3} \rho \AA\cdot\AA' \d\r\,  = +\infty,
	\end{equation*}
	and we obtain 
	\begin{equation*}
	+\infty =\big| \langle \rho \AA, \AA' \rangle \big| \leq 
	\onehalf \langle |\rho|, |\AA|^2 + |\AA'|^2 \rangle.
	\end{equation*}
	Thus $|\AA|^2 + |\AA'|^2 \notin \DensSpace^*$ and either $\vert \AA\vert^2$ or
 $\vert \AA'\vert^2$ is not an element of $\DensSpace^*$ (or both). This contradicts (C1).
\end{proof}

Compatible function spaces can be built up recursively, by combining different function spaces that are already known to be compatible.

\begin{theorem}[Compatibility of intersections and sums] \label{thmRecursiveCompatility}
	Suppose $\DensSpace_1$ and $\CurSpace_1$ are compatible, and the same holds for $\DensSpace_2$ and $\CurSpace_2$. Then (a) the intersections
	\begin{equation}
	    \DensSpace = \DensSpace_1 \cap \DensSpace_2, \quad \CurSpace = \CurSpace_1 \cap \CurSpace_2,
	\end{equation}
	are compatible. Moreover, (b) the sums
	\begin{equation}
	    \DensSpace = \DensSpace_1 + \DensSpace_2, \quad \CurSpace = \CurSpace_1 + \CurSpace_2,
	\end{equation}
	are compatible.
	(Norms of intersections and sums are as given in Definition~\ref{def:norms-union-plus}.)
\end{theorem}
\begin{proof}
	By Theorem~\ref{thmCompatilityPartEquiv} it is sufficient to prove (C1).
	
	Part (a): The dual spaces are $\DensSpace^* = \DensSpace_1^* + \DensSpace_2^*$ and $\CurSpace^* = \CurSpace_1^* + \CurSpace_2^*$. Decompose the vector potential as $\AA = \AA_1 + \AA_2$, with $\AA_1 \in \CurSpace_1^*$ and $\AA_2 \in \CurSpace_2^*$. Using the inequality $$2|\AA_1\cdot\AA_2| \leq  |\AA_1|^2 + |\AA_2|^2$$ property (C1) follows immediately, since $$|\langle \rho, |\AA|^2\rangle| \leq 2 \langle |\rho|, |\AA_1|^2 \rangle + 2\langle |\rho|, |\AA_2|^2 \rangle$$ and each of the two terms is finite by hypothesis.
	
	Part (b): The dual spaces in this case are $\DensSpace^* = \DensSpace_1^* \cap \DensSpace_2^*$ and $\CurSpace^* = \CurSpace_1^* \cap \CurSpace_2^*$. Trivially, for any $\AA \in \CurSpace^*$ we have $\AA \in \CurSpace_i^*$ and therefore, by hypothesis, $|\AA|^2 \in \CurSpace_i^*$ ($i=1,2$). Hence, $|\AA|^2 \in \CurSpace_1^* \cap \CurSpace_2^*$.
\end{proof}

As demonstrated in Ref.~\citenum{Laestadius2014} (see also Theorem~\ref{th:jLp} above), the paramagnetic current density satisfies $\jpvec_\psi \in \vec L^1$ for $\psi\in \mathcal W_N$. Combined with compatibility, this becomes a substantial condition on the vector potential space.

\begin{theorem}
	Let $\DensSpace$ and $\CurSpace$ be compatible function spaces for the particle density and
	current density. 
	Suppose furthermore that $\CurSpace \subseteq \vec L^1$. Then it follows that 
	$$\CurSpace^* \subseteq \vec X^*$$ 
	and from that automatically $\vec \DensSpace \subseteq \CurSpace \subseteq \CurSpace^{**}$. 
\end{theorem}

\begin{proof}
  Let $\vec{A} = (f,g,h) \in Y^*$ be an arbitrary vector potential. From $Y \subseteq \vec L^1$ we immediately have $Y^* \supseteq \vec L^{\infty}$ and therefore we can add a constant to one of the components of $\vec{A}$ without leaving the potential space $Y^*$,
  \begin{equation*}
      \vec{a} = (1+f,g,h) \in Y^*.
  \end{equation*}
  Next, as we assume compatibility of $X$ and $Y$, we have
  \begin{align*}
      \vert \vec a \vert^2 &= (1+f)^2+g^2+h^2 \\
      &= 1 + 2f + f^2+g^2+h^2 \in X^*.
  \end{align*}
  Furthermore, by assumption $f^2,g^2,h^2 \in X^*$. Compatibility combined with the fact that $1 \in L^{\infty}$ yields $1 \in X^*$. Thus, the only remaining term $2f$ must be an element of $X^*$ too. Repeating the proof, with trivial changes for the other components, yields $f,g,h \in X^*$.
\end{proof}

When the preconditions of the above theorem are satisfied and $\AA$ is an allowed vector potential, we thus have the peculiar situation that \emph{both $|\AA|^2$ and $|\AA|$} are contained in the function space of allowed scalar potentials. \newline
\indent Next, we turn to examples of reasonable choices of function spaces for paramagnetic CDFT that also illustrate the duality and compatibility conditions.
Following Lieb~\cite{Lieb83}, we may choose the non-reflexive space
$\DensSpace = L^1 \cap L^3$ for the particle densities and its dual
$\DensSpace^* = L^{3/2} + L^{\infty}$ for the scalar potentials. For
the current densities, we first consider the choice in the literature~\cite{Laestadius2014}, where current densities were placed in the non-reflexive
space $\CurSpace = \vec{L}^1$. Compatibility then follows trivially.

\begin{proposition}
	\label{Prop:VerySimpleCompatibility}
	The choice $\DensSpace = L^1 \cap L^3$ and
	$\CurSpace = \vec{L}^1$ is compatible.
\end{proposition}

As a second example a choice of reflexive, compatible spaces should be given. Reflexivity is imperative for the construction of a well-defined KS scheme like in Sec.~\ref{sec:MYKSODA}. To achieve this we just drop the non-reflexive $L^1$ from $\DensSpace$ in the example above and switch to $\vec{L}^{3/2}$ instead of $\vec{L}^1$ for $\CurSpace$.

\begin{proposition}
	\label{Prop:reflexiveCompatible}
	The strictly convex and reflexive choice $\DensSpace = L^3$,
	$\CurSpace = \vec{L}^{3/2}$ is compatible.
\end{proposition}

Finally the admissible spaces $\DensSpace = L^1 \cap L^3$ and $\CurSpace = \vec{L}^1 \cap \vec{L}^{3/2}$ that were derived from considerations regarding $N$-representability in Sec.~\ref{sec:functionSp} should be examined with respect to compatibility. It turns out that by applying Theorem \ref{thmRecursiveCompatility} we can just build these spaces as intersections of the ones from Propositions \ref{Prop:VerySimpleCompatibility} and~\ref{Prop:reflexiveCompatible}.

\begin{corollary}
	\label{Cor:adminissibleCompatible}
	The choice $\DensSpace = L^1 \cap L^3$ and
	$\CurSpace = \vec{L}^1 \cap \vec{L}^{3/2}$ is compatible.
\end{corollary}

Many other compatible function spaces can be constructed. The previous
examples exclude the common case of uniform magnetic fields as these
require linearly growing vector potentials, e.g., 
$\AA = \frac{1}{2} \vec{B} \times \vec{r}$, which do not belong to any
$\vec L^p(\mathbb{R}^3)$ space. We refer to Ref.~\citenum{Tellgren2018} for a treatment of this situation. The next result provides function spaces
that allow for inclusion of such vector potentials. It is a little detour to other possible choices for compatible Banach spaces, before we return to the discussion of energy functionals defined on them.

Let $w(\r)$ be a suitable weight function. We use the notation $f \in L^p(w)$ for the weighted Lebesgue space defined by
\begin{equation*}
\int_{\mathbb R^3}  |f(\r)|^p w(\r) \d \r < +\infty.
\end{equation*}
Note that $w f \in L^p$ is equivalent to $f \in L^p(w^p)$. 
Some care is required however, as the two forms may produce inequivalent results when multiple weighted $L^p$-spaces are considered.
For example, $f \in L^p(w')\cap L^q(w'')$ in general cannot be represented as $wf \in L^p \cap L^q$, for any weight function $w$. In the following examples the weight function is assumed to
satisfy $w(\r) \geq 1$ for all $\r \in \mathbb{R}^3$.

\begin{theorem}
	Let $Z$ be a normed space with dual $Z^*$.
	Then each of the following choices of function spaces is compatible,
	\begin{equation}\label{eq:thmA1}
	\begin{split}
	& w^2 \rho \in L^1 \cap Z, \quad w \,\jpvec \in \vec L^1,\\
	& w^{-2} v \in L^{\infty} + Z^*, \quad w^{-1} \AA \in \vec L^{\infty}, 
	\end{split}
	\end{equation}
	or, with $1 < p < \infty$,
	\begin{equation}\label{eq:thmA2}
	\begin{split}
	& w^2 \rho \in L^{p/(p-1)} \cap Z, \quad w \,\jpvec \in \vec L^{2p/(2p-1)},\\
	& w^{-2} v \in L^p + Z^*, \quad w^{-1} \AA \in \vec L^{2p}.
	\end{split}
	\end{equation}
\end{theorem}

\begin{proof}
    In both cases, we exploit the equivalence of (C1) and (C2), and only prove one of them.
    First, for the choice in Eq.~\eqref{eq:thmA1}, condition (C1) is trivial, since $w^{-1} \AA \in \vec L^{\infty}$ directly implies $w^{-2} |\AA|^2 \in L^{\infty}$. Second, for Eq.~\eqref{eq:thmA2}, we similarly have that $w^{-1} \AA \in \vec L^{2p}$ directly implies $w^{-2} |\AA|^2 \in L^{p}$, establishing (C1).
\end{proof}

In the first of the above examples, Eq.~\eqref{eq:thmA1}, we can make the trivial choice $w(\r) = 1$ and $Z =  L^3$ to recover the function space analyzed in Ref.~\citenum{Laestadius2014}. Moreover, with suitable, non-trivial weight functions, unbounded vector potentials can be considered. For example, the choice $w(\r) = (1+|\r|^2)^{1/2}$ was studied extensively in Ref.~\citenum{Tellgren2018} as it allows for uniform magnetic fields and always ensures that the angular momentum is well-defined.

\subsection{Finiteness of the energy functional}
\label{sec:Efiniteness}

The following general property of the CDFT energy functional $E^\lambda(v,\AA)$ will be useful later. It says that for both interacting and non-interacting systems the energy is finite for all considered potentials. Furthermore, if the choice of function spaces is compatible, the same boundedness from below holds for $\bar{E}^\lambda(u,\AA)$. 

\begin{lemma}\label{lemma:Efiniteness}
$E^\lambda(v,\AA)$ is finite for
\[
(v,\AA) \in (L^{3/2} + L^\infty) \times ( \vec L^3 + \vec L^\infty),
\]
corresponding to the density space $(\rho,\jpvec) \in (L^1\cap L^3) \times (\vec L^1 \cap \vec L^{3/2})$.
By compatibility, $\bar E^\lambda(u,\AA)$ is also finite on the same domain.
\end{lemma}

\begin{proof}
To prove finiteness (of the infimum) it is enough to prove boundedness from below. 
By definition, we have for any $(v,\AA)\in (L^{3/2} + L^\infty) \times ( \vec L^3 + \vec L^\infty)$
\begin{align}
E^\lambda(v,\AA) &=\inf_{\substack{(\rho,\jpvec) \\ \in (L^1\cap L^3) \times (\vec L^1 \cap \vec L^{3/2})  }}\Big\{ F^\lambda(\rho,\jpvec) +
\int_{\mathbb R^3} \AA \cdot \jpvec \d \r \nonumber \\
& \quad  \quad +  \int_{\mathbb R^3} \Big(v + \onehalf\vert \AA \vert^2 \Big) \rho \d \r \Big\}.
\label{eq:EvAA1}
\end{align}
Moreover, if $(\rho,\jpvec)$ is not $N$-representable then $F^\lambda(\rho,\jpvec) =+\infty$ and we can consequently consider only $N$-representable density pairs.
By definition $F^\lambda \geq F^0$ and furthermore $F^0(\rho,\jpvec) \geq  \mathrm{inf}_\psi\, \mathcal K(\psi)$ for all $\psi$ with $\rho_\psi= \rho$ and $\jpvec_\psi = \jpvec$. 
By the von Weizs\"acker bound in Eq.~\eqref{eq:Furches}, we obtain  
\[
F^0(\rho,\jpvec)  \geq \frac 1 2\int_{\mathbb R^3} \vert \nabla \rho^{1/2} \vert^2 \d \r +\frac 1 2 \int_{\mathbb R^3} \vert \jpvec \vert^2\rho^{-1} \d \r,
\]
for $N$-representable density pairs $(\rho,\jpvec)$. The Sobolev inequality in Eq.~\eqref{Sob} further allows for
\begin{align}
F^0(\rho,\jpvec) \geq \frac{1}{2S} \Vert \rho \Vert_3 + \frac 1 2 \int_{\mathbb R^3}  \vert \jpvec \vert^2 \rho^{-1} \d \r.
\label{eq:EvAA2}
\end{align}
Combining Eqs.~\eqref{eq:EvAA1} and \eqref{eq:EvAA2}, it follows
\begin{align*}
E^\lambda(v,\AA) \geq & \inf_{(\rho,\jpvec) } \Big\{ \frac 1 {2S} \Vert \rho \Vert_3 + \frac 1 2 \int \vert \jpvec\vert^2 \rho^{-1}\d\r \\ &+
\int_{\mathbb R^3} \AA \cdot \jpvec \d \r 
+ \int_{\mathbb R^3} \Big(v +\onehalf \vert \AA \vert^2 \Big) \rho \d \r \Big\},
\end{align*}
where the minimization is restricted to $N$-representable density pairs. Using the obvious inequality $$\vert \jpvec/\rho^{1/2} + \rho^{1/2} \AA \vert^2\geq 0$$ to replace the full square, we have 
\begin{align}
E^\lambda(v,\AA) \geq \inf_{\rho \in\mathcal I_N } \Big\{  \frac{1}{2S} \Vert \rho\Vert_3  + \int_{\R^3} v \rho \d \r \Big\}.
\label{eq:EvAA3}
\end{align}

Next we bound the r.h.s.\ of Eq.~\eqref{eq:EvAA3} from below. Since $\mathcal{C}_0^\infty$ is dense in $L^{3/2}$, also $L^\infty$ is. This means that we can split $v = v_{S} + v_\infty$ with $v_{S} \in L^{3/2}$, $v_\infty \in L^\infty$ in such a way that $\|v_{S}\|_{3/2}$ is arbitrarily small. We choose the decomposition such that $\|v_{S}\|_{3/2} \leq (2S)^{-1}$ with the constant $S$ from Eq.~\eqref{Sob2}. Then using H\"older's inequality Eq.~\eqref{eq:EvAA3} can be estimates as
\begin{align*}
E^\lambda(v,\AA) &\geq
\inf_{\rho \in \mathcal I_N} \Big\{  \Big( \frac{1}{2S}  - \|v_{S}\|_{3/2} \Big) \|\rho\|_3 + \int_{\R^3} v_\infty \rho \d \r \Big\} \\
&\geq -N\|v_\infty\|_\infty.
\end{align*}
This proves the claim for $E^\lambda(v,\AA)$ and by compatibility also for $\bar E^\lambda(u,\AA)$.
\end{proof}

\begin{remark}\label{rmk:finitXRYR}
The reflexive space $X_\mathrm{R}^* \times Y_\mathrm{R}^*$ is a subset of the domain given in Lemma~\ref{lemma:Efiniteness} and corresponds to a compatible density space by Proposition~\ref{Prop:reflexiveCompatible}. It follows by an adaptation of the proof of Lemma~\ref{lemma:Efiniteness} that $E^\lambda(v,\AA)$ and $\bar E^\lambda(u,\AA)$ are finite on $X_\mathrm{R}^* \times Y_\mathrm{R}^*$.
In this case when we just assume $(\rho,\jpvec) \in L^3\times \vec L^{3/2}$, however, it is crucial to exploit that the minimization can be restricted to $N$-representable density pairs, such that 
\begin{align*}
\int_{\R^3} v\rho \d\r &=  \int_{\R^3} (v - \varphi)\rho \d \r+ \int_{\R^3} \varphi \rho \d\r \\
&\geq -\Vert v - \varphi \Vert_{3/2} \Vert \rho\Vert_3 - N \Vert \varphi\Vert_\infty,
\end{align*}
where $\varphi \in C_0^\infty$ has been chosen such that $v - \varphi$ has sufficiently small $L^{3/2}$ norm.
\end{remark}

\section{Regularization and the Kohn--Sham iteration scheme}
\label{sec:IIIks}
In our previous work~\cite{KSpaper2018} the general theory of a quantum system described by the (density) variable $b\in B$, where $B$ is a reflexive ($B^{**} = B$) and strictly convex Banach space, was presented. In this theory the general problem
\[
E(b^*) = \inf_{b\in B} \{ f(b) + \langle b^*,b\rangle \}
\]
(here $\langle b^*,b \rangle$ denotes the dual pairing that is not necessary given by an integral) with $f\in \Gamma_0(B)$ is studied.
We here wish to apply this structure to paramagnetic CDFT, i.e., $b= (\rho,\jpvec)$, $b^*=(u,\AA)$, and $f(b) = F^\lambda(\rho,\jpvec)$. We choose the density space from Definition~\ref{def:XRYR}
\[
X_\mathrm{R}\times Y_\mathrm{R} = L^3 \times \vec L^{3/2}
\]
that was already discussed in Proposition \ref{Prop:reflexiveCompatible} to meet the requirements of strict convexity, reflexivity, and compatibility. The dual potential space then is 
\[
X_\mathrm{R}^* \times Y_\mathrm{R}^* = L^{3/2} \times \vec L^{3}.
\]
Another way to gain reflexivity is by limiting the spatial domain to a bounded set $\Omega \subsetneq \R^3$. Then $\DensSpace = L^{1}(\Omega) \cap L^{3}(\Omega),\CurSpace = \vec L^{1}(\Omega) \cap \vec L^{3/2}(\Omega)$ automatically collapse to $\DensSpace = L^{3}(\Omega),\CurSpace = \vec L^{3/2}(\Omega)$, which are again reflexive. We also want to remark that in this case the usual Coulomb potential is included in $L^{3/2}(\Omega)$ since it is in $L_\mathrm{loc}^{3/2}(\R^3)$.
\newline
\indent 
Furthermore, compatibility of $X_\mathrm{R} \times Y_\mathrm{R}$ gives that the concave $\bar E^\lambda$ (restricted to $X_\mathrm{R}^* \times Y_\mathrm{R}^*$) can be defined. By Lemma~\ref{lemma:Efiniteness} and Remark~\ref{rmk:finitXRYR}, this energy is also bounded below. 
To connect the CDFT functionals $F_\mathrm{VR}^\lambda(\rho,\jpvec)$ (or $F_\mathrm{VR,DM}^\lambda(\rho,\jpvec)$) and $\bar E^\lambda(u,\AA)$, we have already introduced the Legendre--Fenchel transformation in Eqs.~\eqref{eq:F-LF} and \eqref{eq:E-LF}. It also relates the Lieb functional $F^\lambda$ and the concave energy functional $\bar E^\lambda$ vice versa as a conjugate pair. Also note that $F^\lambda \in \Gamma_0(X_\mathrm{R} \times Y_\mathrm{R})$ by Lemma~\ref{lemma:Gamma}.
The next step lies in another type of transformation that makes $F^\lambda$ functionally differentiable too, which is achieved by the Moreau--Yosida regularization.

\subsection{Moreau--Yosida regularization}
\label{sec:MY}

The original problem given in Eq.~\eqref{eq:E-LF} of finding a ground-state density pair $(\rho,\jpvec) \in X_\mathrm{R} \times Y_\mathrm{R}$ by minimizing the \emph{convex} functional $F^\lambda$ plus the potential energy can be restated using sub-/superdifferentials (Definition~2 in Ref.~\citenum{KSpaper2018}), both denoted $\partial$.
It means selecting $(\rho,\jpvec)$ from the superdifferential of $\bar E^\lambda$ at $(u,\AA)\in X_\mathrm{R}^* \times Y_\mathrm{R}^*$. Through the Legendre--Fenchel transformation Eq.~\eqref{eq:F-LF}, the same is possible for $F^\lambda$. With a minus sign in front, the potential pair $(u,\AA)$ yielding the ground state lies in the subdifferential of $F^\lambda$ at $(\rho,\jpvec)$ (see Lemmas~3 and~4 in Ref.~\citenum{KSpaper2018}),
\begin{equation}\label{eq:conjugate-differentials}
	(\rho,\jpvec) \in \partial \bar E^\lambda(u,\AA)  \Longleftrightarrow  -(u,\AA) \in \partial F^\lambda(\rho,\jpvec).
\end{equation}
This statement can be seen as a more general reformulation of the HK theorem, but only with a $(u,\AA)$ potential pair, which is different from the physical $(v,\AA)$ setting.\newline
\indent The generalized notions of differentiability for convex/concave functionals involve the difficulty of non-existence or non-uniqueness. Sub- and superdifferentials are set-valued and can thus be empty or contain many elements. It is thus beneficial to ``smooth out'' the functional $F^\lambda$ in such a way that it is differentiable, which implies that the subdifferential contains only one single element. (Note that in infinite dimensions only for a continuous functional a single element in the subdifferential implies differentiability.) 
This is achieved by the Moreau--Yosida regularization of $F^\lambda \in \Gamma_0(X_\mathrm{R} \times Y_\mathrm{R})$, for $\eps>0$ given by
\begin{equation}\label{eq-defMY}
F^\lambda_\eps(\rho,\jpvec) = \inf_{\substack{(\sigma,\vec k) \\ \in X_\mathrm{R} \times Y_\mathrm{R}}} \left\{ F^\lambda(\sigma,\vec k) + \frac{1}{2\eps} \Vert (\rho,\jpvec)-(\sigma,\vec k)\Vert^2 \right\}.
\end{equation}
Since $F^\lambda$ is convex, the new functional $F_\eps^\lambda$ is convex as well and now also functionally differentiable by Theorem~9 in Ref.~\citenum{KSpaper2018}.
This regularized functional then serves as the basis for defining a new energy functional $\bar E^\lambda_\eps$ through the Legendre--Fenchel transformation Eq.~\eqref{eq:E-LF} again
\begin{align*}
\bar{E}_\eps^{\lambda}(u,\AA) = \inf_{(\rho,\jpvec) \in X_\mathrm{R} \times Y_\mathrm{R}} \Big\{ F_\eps^{\lambda}(\rho,\jpvec) + \langle \AA, \jpvec \rangle + \langle u,\rho\rangle \Big\}.
\end{align*}
Note carefully that $\bar{E}_\eps^{\lambda}$ is not the Moreau--Yosida regularization of some functional, but instead the Legendre--Fenchel conjugate of the regularized functional $F_\eps^\lambda$. Then Theorem~10 in Ref.~\citenum{KSpaper2018} can be used to relate the two energy functionals through
\begin{equation}\label{eq:E-back}
\begin{aligned}
\bar{E}^{\lambda}(u,\AA) &= \bar{E}_\eps^{\lambda}(u,\AA) + \frac{\eps}{2} \Vert (u,\AA) \Vert^2.
\end{aligned}
\end{equation}
\indent Since the infimum in the definition Eq.~\eqref{eq-defMY} is always uniquely attained at some $(\rho_\eps,\jpvec_\eps) \in X_\mathrm{R}\times Y_\mathrm{R}$ (see Sec.~2.2.3 in Ref.~\citenum{barbu2012convexity}), we can define a mapping $(\rho,\jpvec) \mapsto (\rho_\eps,\jpvec_\eps)$ that is called the proximal mapping, i.e., 
\begin{equation}\label{eq:prox}
(\rho_\eps,\jpvec_\eps) = \mathrm{prox}_{\eps F}(\rho,\jpvec)
\end{equation}
and
\begin{align*}
F_\eps^\lambda(\rho,\jpvec)  =F^\lambda(\rho_\eps, \jpvec_\eps ) + \frac 1 {2\eps} \Vert (\rho,\jpvec) - (\rho_\eps,\jpvec_\eps )  \Vert^2.
\end{align*}
The proximal mapping maps density pairs that are solutions of the regularized problem $(\rho,\jpvec) \in \partial \bar{E}_\eps^{\lambda}(u,\AA)$ back to solutions of the corresponding unregularized problem, $\mathrm{prox}_{\eps F} (\rho,\jpvec) \in \partial \bar{E}^{\lambda}(u,\AA)$ by Corollary~11 in Ref.~\citenum{KSpaper2018}.
Furthermore, the original functional $F^\lambda$ is subdifferentiable at $(\rho_\eps,\jpvec_\eps)\in X_\mathrm{R} \times Y_\mathrm{R}$, which means that the density pair $(\rho_\eps,\jpvec_\eps)$ is $v$-representable. We note that by Theorem~9 in Ref.~\citenum{KSpaper2018}
\begin{align*}
\nabla F_\eps^\lambda\xp =\eps^{-1} \mathcal J(\rho-\rho_\eps,\jpvec-\jpvec_\eps)\in \partial F^\lambda(\rho_\eps,\jpvec_\eps),
\end{align*}
where $\mathcal J: X_\mathrm{R} \times Y_\mathrm{R} \to X_\mathrm{R}^* \times Y_\mathrm{R}^*$ is the duality map that is always homogeneous (Definition~7 in Ref.~\citenum{KSpaper2018}). By Proposition~1.117 in Ref.~\citenum{barbu2012convexity}, it is further bijective in the present setting of reflexive and strictly convex Banach spaces (including their duals). Letting $\mathcal J = (\mathcal J_{X_\mathrm{R}},\mathcal J_{Y_\mathrm{R}})$, we get from $-(u,\AA)=\nabla F_\eps^\lambda\xp$ that
\begin{align*}
u & =- \eps^{-1}\mathcal J_{X_\mathrm{R}}(\rho -\rho_\eps)\in X_\mathrm{R}^*,\\
\AA &= - \eps^{-1} \mathcal J_{Y_\mathrm{R}}(\jpvec -\jpvec_\eps)\in Y_\mathrm{R}^*,
\end{align*}
which straightforwardly transforms to
\begin{equation}
\begin{aligned}\label{eq:dual-map}
\rho & =\rho_\eps - \eps\mathcal J_{X_\mathrm{R}}^{-1}(u)\in X_\mathrm{R},\\
\jpvec &= \jpvec_\eps - \eps \mathcal J_{Y_\mathrm{R}}^{-1}(\AA)\in Y_\mathrm{R}.
\end{aligned}
\end{equation}
Here the compatibility of $X_\mathrm{R}$ as given by Proposition~\ref{Prop:reflexiveCompatible} again becomes important since it implies that we can decompose $u$ as
\begin{equation}\label{eq:decompose-u}
u = \Big(u - \onehalf \vert \AA \vert^2\Big) + \onehalf \vert \AA \vert^2=: v + \onehalf |\AA|^2.
\end{equation}
Recall that for all $u \in L^{3/2}$ and $\AA\in\vec L^3$ we have
\[
u - \onehalf \vert \AA \vert^2 \in L^{3/2}.
\]

We conclude this section by discussing a Hilbert (H) space formulation.  
A direct adaptation of the approach taken in Ref.~\citenum{Kvaal2014} is to choose  
\[
X_\mathrm{H} \times Y_\mathrm{H}= L^2\times \vec L^2,
\]
i.e., the Hilbert space built up from four copies of $L^2$. The regularization presented in Ref.~\citenum{Kvaal2014} for standard DFT can then be directly applied to the four-vector $(\rho,\jpvec)$ instead of just the particle density $\rho$. We note that $\psi\in \mathcal{H}^2$ is sufficient to obtain $j^k\in L^2$, see Proposition~\ref{prop:L2}.
However, when $X\times Y=X_\mathrm{H} \times Y_\mathrm{H}$ the density and potential spaces are not compatible in the meaning of Definition~\ref{defCompatibility}. This causes a problem for the regularization procedure of the Legendre--Fenchel pair $F^\lambda$ and $\bar E^\lambda$, because we cannot decompose $u$ as in Eq.~\eqref{eq:decompose-u} any more. To see this, note that $\vert \AA \vert^2$ cannot in general be assumed to satisfy $|\AA|^2 \in X_\mathrm{H}$, which in turn would yield the desired $v \in X_\mathrm{H}$.
Consequently, we cannot obtain the physical setting of $E^\lambda(v,\AA)$ from the Moreau--Yosida setting of $\bar E^\lambda(u,\AA)$. 
This again highlights the usefulness of the more general reflexive Banach space formulation that allows a compatible choice of function spaces and makes a regularized paramagnetic CDFT possible.

\subsection{Regularized Kohn--Sham iteration scheme in CDFT}
\label{sec:KS}

We now revisit the KS iteration scheme that we previously analyzed for generic Banach spaces~\cite{KSpaper2018}. Due to the similarity to the Optimal Damping Algorithm~\cite{Cances2000,CANCES_JCP114_10616} constructed for an unregularized setting, we baptize this iteration scheme the Moreau--Yosida Kohn--Sham Optimal Damping Algorithm (MYKSODA).

Again, let $X \times Y = X_\mathrm{R} \times Y_\mathrm{R}$ so that the space of densities is compatible, reflexive, and strictly convex. The latter two properties are also fulfilled by the dual $X^*\times Y^*= X_\mathrm{R}^* \times Y_\mathrm{R}^*$.
From Sec.~\ref{sec:MY}, the ground-state problem Eq.~\eqref{eq:E-LF} can be reformulated in terms of sub- and superdifferentials
\begin{align*}
	(\rho,\jpvec) \in \partial \bar E^\lambda(u,\AA)   \Longleftrightarrow  -(u,\AA) \in \partial F^\lambda(\rho,\jpvec).
\end{align*}
The regularized functionals $F_\eps^\lambda$ and $\bar E_\eps^\lambda$ then allow the same relation with the benefit that $F_\eps^\lambda$ is now differentiable. This means we can switch from the subdifferential $\partial$ to the gradient $\nabla$ of $F_\eps^\lambda$, yet this is not permitted for $\bar E_\eps^\lambda$. \newline
\indent We set up two problems side by side, the interacting problem with $\lambda=1$ and the non-interacting reference problem corresponding to $\lambda=0$, i.e.,
\begin{align}
	\xreg &\in \partial \bar E_\eps^1\uExt \nonumber\\
	&\Longleftrightarrow  -\uExt = \nabla F_\eps^1\xreg, \label{eq:reg-problem-1} \\
	\xreg &\in \partial \bar E_\eps^0\uKS \nonumber \\
	&\Longleftrightarrow  -\uKS = \nabla F_\eps^0\xreg. \label{eq:reg-problem-0}
\end{align}
In the setting of (regularized) KS theory the external potential pair $\uExt$ is fixed and $\uKS$ is to be determined under the assumption that both problems give the same (regularized) ground-state density pair $\xreg$. We have highlighted the $\lambda$ dependence by using ``ext'' for $\lambda=1$ and by ``KS'' for $\lambda=0$ (but for the functionals $\bar E$ and $F$ we keep 1 and 0). By combining Eqs.~\eqref{eq:reg-problem-1} and \eqref{eq:reg-problem-0} we arrive at
\begin{equation*}
\begin{aligned}
	\uKS =\; &\uExt \\
	&+ \nabla F_\eps^1\xreg - \nabla F_\eps^0\xreg
\end{aligned}
\end{equation*}
from which the iteration scheme will be derived by replacing the unknown variables by sequences.
Let $(\rho_i,\vec{j}_i)$ be the element of a sequence towards the (regularized) ground-state density pair $(\rho_{\mathrm{reg}},\jpvec_{\mathrm{reg}})$, and thus the next step towards the KS potential pair $\uKS$ follows by
\begin{equation}\label{eq:KS1}
\begin{aligned}
	(u_{i+1},\AA_{i+1}) =\; &\uExt \\
	&+ \nabla F_\eps^1 \xpi - \nabla F_\eps^0 \xpi.
\end{aligned}
\end{equation}
The expression $\nabla (F_\eps^1 - F_\eps^0)$, which can be identified as a ``Hartree exchange-correlation potential'', is where the usual approximation techniques of DFT enter. Yet in the domain of CDFT the variety of tried and tested functionals is meager \cite{vignale1990g,tellgren2014non,furness2015current} compared to the wealth of options in conventional DFT \cite{burke2012perspective}.\newline
\indent The second major step in the iteration scheme is then the solution of the non-interacting reference system (instead of the computationally difficult interacting problem) which selects a ground-state density pair $(\rho_{i+1},\jpvec_{i+1})$ corresponding to the approximated KS potential pair $(u_{i+1},\AA_{i+1})$ that then serves as the next input in Eq.~\eqref{eq:KS1}
\begin{equation}\label{eq:KS2}
	(\rho_{i+1},\jpvec_{i+1}) \in \partial \bar E_\eps^0(u_{i+1},\AA_{i+1}).
\end{equation}
That this (super)differential is indeed always non-empty, meaning that the associated ground-state problem has at least one solution, can be shown by using the result of finiteness of $\bar E^\lambda$ from Lemma \ref{lemma:Efiniteness} (see the proof of Theorem 12 in Ref.~\citenum{KSpaper2018}). \newline
\indent The iteration stops when 
\[
\uExt = -\nabla F_\eps^1\xpi,
\]
which means that $\xpi$ solves the original interacting ground-state problem with external potential pair $\uExt$. This is so because this condition is equivalent to 
\[
(u_{i+1},\AA_{i+1}) = -\nabla F_\eps^0 \xpi
\]
by Eq.~\eqref{eq:KS1}, and the next step given by Eq.~\eqref{eq:KS2} would just yield the same density pair $(\rho_{i+1},\jpvec_{i+1})=\xpi$ again.
In such a case, or if we decide we have converged close enough to the supposed ground-state density pair $(\rho_{\mathrm{reg}},\jpvec_{\mathrm{reg}})$ of the regularized problem, a fixed relation between this solution and the solution of the unregularized problem is established by Eq.~\eqref{eq:dual-map}. \newline
\indent The question of convergence of the sequences $\{ \xpi\}$ and $\{(u_i,\AA_i)\}$ with respect to the Banach space topologies of $X_\mathrm{R}\times Y_\mathrm{R}$ and  $X_\mathrm{R}^*\times Y_\mathrm{R}^*$ is immediately raised. The authors have answered this in Ref.~\citenum{KSpaper2018}, but only in a weak sense. 
More precisely, the associated energy sequence 
\[
F^1_\eps \xpi + \langle \AA_\mathrm{ext},\jpvec_i \rangle + \langle u_\mathrm{ext},\rho_i\rangle
\]
can be guaranteed to converge to \emph{some} value larger or equal the correct value of the regularized energy functional $\bar E^1_\eps \uExt$. Arguably this is not what one expects from a well-formed KS iteration, where convergence to the correct ground-state density pair is the obvious aim. Further, such convergence in terms of energy can only be guaranteed if an additional step is inserted into the scheme consisting of Eqs.~\eqref{eq:KS1} and \eqref{eq:KS2}, coined ``optimal damping'' \cite{Cances2000,Cances2000b,Cances2001}, that limits the step of the new density in such a way that the energy value assuredly decreases.
It should be noted that a very recent development~\cite{penz2019guaranteed} treating the finite dimensional case finally proves full convergence of the exact KS iteration.

\subsection{Weak-type convergence of MYKSODA}
\label{sec:MYKSODA}

We now have all the ingredients we need for an application of Theorem~12 in Ref.~\citenum{KSpaper2018}, which constitutes the main theoretical result of this work.

\begin{theorem}\label{th:KS}
For the density spaces $X_\mathrm{R} \times Y_\mathrm{R}  = L^3 \times \vec L^{3/2}$, and the corresponding dual for potential pairs $X_\mathrm{R}^* \times Y_\mathrm{R}^* = L^{3/2} \times \vec L^3$, a well-defined KS iteration can be set up for the energy functional $\bar E^\lambda$ from Eq.~\eqref{eq:E-functional}. It starts with a fixed potential pair $\uExt \in X_\mathrm{R}^* \times Y_\mathrm{R}^*$ by setting $(u_1,\AA_1) = \uExt$ and selecting $(\rho_1,\jpvec_1) \in \partial \bar E^0_\eps\uExt$. Then iterate $i=1,2,\ldots$ according to:
\begin{framed}  
	\begin{enumerate}[(a),leftmargin=.65cm]
		\item Set 
		\begin{align*}
		(u_{i+1},\AA_{i+1}) = \;&\uExt \\&+ \nabla F_\eps^1\xpi - \nabla F_\eps^0 \xpi
		\end{align*}
		and stop if $$(u_{i+1},\AA_{i+1}) = -\nabla F_\eps^0\xpi = \uKS.$$
		\item Select $ (\rho_{i+1}',\jpvec_{i+1}') \in \partial \bar E_\eps^0(u_{i+1},\AA_{i+1})$.
		\item Choose $t_i\in (0,1]$ maximally such that for
		\begin{align*}
		(\rho_{i+1},\jpvec_{i+1}) =& \xpi \\ &+ t_i\big(   (\rho_{i+1}',\jpvec_{i+1}') - \xpi  \big)    
		\end{align*}
		one still has
		\begin{equation}\label{eq:oda-step}
		\begin{aligned}
		\frac{\mathrm{d}}{\mathrm{d} t_i}&\left[ F_\eps^1(\rho_{i+1},\vec j_{i+1}) + \langle \uExt, (\rho_{i+1},\vec j_{i+1}) \rangle \right] \\
		&= \langle \nabla F^1_\eps(\rho_{i+1},\jpvec_{i+1}) +\uExt, \\
		&\quad\quad (\rho_{i+1}',\jpvec_{i+1}') - (\rho_i,\jpvec_i) \rangle \leq 0. 
		\end{aligned}
		\end{equation}
	\end{enumerate}
\end{framed}

	Then the strictly descending sequence 
	\begin{align*}
	\big\{ F^1_\eps \xpi + \langle \AA_\mathrm{ext},\jpvec_i\rangle + \langle u_\mathrm{ext},\rho_i\rangle \big\}_i 
	\end{align*}
	converges as a sequence of real numbers to
	\begin{align*}
	e_\eps\uExt &= \inf_i \big\{ F^1_\eps \xpi + \langle \AA_\mathrm{ext},\jpvec_i\rangle + \langle u_\mathrm{ext},\rho_i\rangle \big\} \\
	            & \geq \bar E^1_\eps\uExt.
	\end{align*}
	Thus,
	\[
	e_\eps\uExt + \frac{\eps}{2} \Vert (u_\mathrm{ext}, \AA_\mathrm{ext}) \Vert^2
	\]
	is an upper bound for the ground-state energy $\bar E^1 \uExt$.
\end{theorem}

\begin{proof}
To be able to apply Theorem 12 from Ref.~\citenum{KSpaper2018} we have to make sure that $X_\mathrm{R} \times Y_\mathrm{R}$ and $X_\mathrm{R}^* \times Y_\mathrm{R}^*$ are reflexive and strictly convex, the non-interacting energy functional $\bar E^0$ needs to be finite on all of $X_\mathrm{R}^* \times Y_\mathrm{R}^*$, and $F^\lambda,\bar E^\lambda$ must form a convex-concave pair linked by the Legendre--Fenchel transformation. Now, Proposition \ref{Prop:reflexiveCompatible} shows that the chosen density spaces are indeed reflexive, all $L^p$ with $1<p<\infty$ are strictly convex anyway, and also that they are compatible. Compatibility gives that the energy functional can be transformed to a concave $\bar E^\lambda$ by Eq.~\eqref{eq:E-transform} that links to a convex Lieb functional $F^\lambda$ by Eq.~\eqref{eq:F-LF} (Legendre--Fenchel transformation). Then Lemma~\ref{lemma:Efiniteness} and Remark~\ref{rmk:finitXRYR} prove that $\bar E^0$ is indeed finite on $X_\mathrm{R}^* \times Y_\mathrm{R}^*$. 
With the results from Theorem 12 in Ref.~\citenum{KSpaper2018} we get a strictly decreasing and converging sequence 
\[
\big\{ F^1_\eps \xpi + \langle \AA_\mathrm{ext},\jpvec_i\rangle + \langle u_\mathrm{ext},\rho_i\rangle \big\}_i 
\]
with the given lower bound. The transformed energy bound follows directly from Eq.~\eqref{eq:E-back}.
\end{proof}

\begin{remark}
Any candidate for a possible ground-state density pair from the iteration defined in Theorem~\ref{th:KS} can be transformed to a solution of the corresponding ``physical'' unregularized problem with the help of Eq.~\eqref{eq:dual-map}. But it has not been proven that the iteration actually converges in terms of densities and potentials as elements of the given Banach spaces and dual spaces \emph{or} that if it converges, it actually reaches the ground-state density pair $\xreg$ and the associated KS potential pair $\uKS$.
\end{remark}

\begin{remark}
We already discussed $\uExt = -\nabla F_\eps^1 \xpi$ as a stopping condition for the iteration before. To have such a stopping condition is important for at least a possible convergence to the correct ground-state density pair that is then a fixed point. Note that we still have the appearance of a whole set of ground-state density pairs in step (b), signifying that degeneracy is admitted. Thus, it is more beneficial to look at the sequence of potential pairs $(u_i,\AA_i)$ that, if the stopping condition is eventually reached, gives the correct KS potential for \emph{some} ground-state densities. This is a reason why it was important to switch to differentiable functionals $F_\eps^\lambda$ through regularization: To be able to define a unique sequence of potentials that can converge to the KS potential. See also Ref.~\citenum{LammertBivariate} where the traditional iteration in density space is supplemented with a \emph{bivariate} formalism.
\end{remark}

\begin{remark}
Note that step (c) in the KS iteration scheme above corresponds to a line-search between the points $\xpi$ and $(\rho_{i+1}',\jpvec_{i+1}')$. If more solutions of the regularized, non-interacting reference system from step (b) are taken into account (degeneracy), then step (c) gets generalized to a search over a convex polytope.
\end{remark}

\subsection{Kohn--Sham iteration scheme for two-electron systems}
\label{sec:KSconvergence}

As only a partial convergence result is presently available for the KS algorithm---leaving open the possibility that it does not always converge to the right energy and potential---it is interesting to consider a case where non-interacting $N$- and $v$-representability issues pose a challenge for the algorithm.
Consider a formulation of paramagnetic CDFT for singlet ground states and restrict attention to a two-electron system. Then the (unregularized) non-interacting KS system is represented by a single orbital and its vorticity,
\begin{equation*}
\vec{\nu}_{\mathrm{KS}} =
\nabla\times\frac{\vec{j}_{\mathrm{KS}}}{\rho_{\mathrm{KS}}},
\end{equation*}
vanishes if differentiability is assumed. For the small set of KS potentials that yield ground-state degeneracies, this can be circumvented by allowing the KS system to be represented by a mixed state. However, for most KS potentials, the KS ground state is unique and has trivial vorticity. On the other hand, for most external potentials, the correlated ground state of the interacting system has a nontrivial vorticity. Hence, most of the ground-state densities are not non-interacting ($\lambda=0$) $N$-representable.

This situation poses an interesting challenge for the KS iteration scheme as the vorticity of the KS system cannot develop gradually. Until the algorithm has constructed potentials that yield an exact ground-state degeneracy, the vorticity is trivial, and it is not clear how ``visible'' the corresponding degrees of freedom are to the optimization algorithm. Moreau--Yosida regularization alleviates the challenge somewhat, since the relevant densities have contributions from the potentials, complicating the non-interacting $N$-representability conditions. With regularization, the relevant vorticity that should reproduce the interacting system using Eq.~\eqref{eq:dual-map} is
\begin{equation*}
    \vec{\nu}_{\mathrm{KS,reg}} =
\nabla\times\frac{\vec{j}_{\mathrm{KS}} - \eps \mathcal{J}_{Y_\mathrm{R}}^{-1}(\AA^{\mathrm{KS}})}{\rho_{\mathrm{KS}} - \eps \mathcal{J}_{X_\mathrm{R}}^{-1}(u^{\mathrm{KS}})}.
\end{equation*}
If counterexamples that prevent a full convergence proof exist at all, the type of system sketched above is a promising candidate for further analysis.

\section{Numerical application to quantum ring}
\label{sec:numerApp}

The theory of regularized CDFT can be directly applied to a one-dimensional quantum ring. Although this is a toy model, it is sufficiently rich to contain simple formal analogues of many aspects of CDFT for a three-dimensional spatial domain. We limit attention to two-electron systems in singlet spin states. The Hamiltonian is given by
\begin{align*}
    H^{\lambda}(v,A) &= \frac{1}{2} \sum_{j=1}^2 \Big(-\frac{\i}{R} \frac{\partial}{\partial \theta_j} + A(\theta_j) \Big)^2 \\
    &\quad + \sum_{j=1}^2 v(\theta_j) + \lambda W(\theta_1,\theta_2),
\end{align*}
where $R$ is the radius of the ring and the potentials $v$ and $A$ as well as the electron-electron repulsion $W$ are considered functions of the angular position along the ring. Note that gradients and the vector potential only have tangential components and may therefore effectively be treated as scalars. \newline 
\indent Because of the limitation to singlet states, the spatial wave function $\psi(\theta_1,\theta_2) = \psi(\theta_2,\theta_1)$ must be symmetric. Any uncorrelated state, e.g., a KS state, takes the form $\psi(\theta_1,\theta_2) = \phi(\theta_1) \phi(\theta_2)$ and is defined by a single orbital $\phi$. The densities that arise from such an uncorrelated, single-orbital $\phi(\theta) = \sqrt{\rho(\theta)} \, e^{i\chi(\theta)}$ state must satisfy
\begin{equation*}
    R \int_0^{2\pi} \frac{j(\theta)}{\rho(\theta)} \d\theta = \int_0^{2\pi} \nabla\chi(\theta) \d\theta = 2\pi m,
\end{equation*}
with $m$ an integer if $\rho>0$ everywhere. Note that $\chi(\theta+2\pi) = \chi(\theta) + 2\pi m$ is in general a multivalued phase function.
By contrast, a correlated state can well give rise to a fractional value of $m$. This is the quantum ring analogue of the fact that vorticity is trivial for single-orbital systems in a three-dimensional spatial domain. \newline
\indent Next, in order to study regularized CDFT numerically, we discretize the quantum ring into $N_G$ uniformly spaced grid points. The approach described below is implemented in a Matlab program named {\sc MYring}~\cite{MYringProgram}. We replace the Laplacian by the standard second-order finite difference expression
\begin{equation*}
    \frac{\partial^2}{\partial \theta^2} \phi(\theta) \leftrightarrow \frac{\phi(\theta_{j+1}) - 2\phi(\theta_j) + \phi(\theta_{j-1})}{h^2},
\end{equation*}
where $h = 2\pi R/N_G$ is the grid spacing. The paramagnetic term is discretized using the symmetric first-order expression,
\begin{equation*}
    \frac{\partial}{\partial \theta} \phi(\theta) \leftrightarrow \frac{\phi(\theta_{j+1}) - \phi(\theta_{j-1})}{2h}.
\end{equation*}
Defining the particle density and current density at grid point $k$ by
\begin{align*}
    \rho_k & = h \sum_l |\psi(\theta_k,\theta_l)|^2,
            \\
    j_k & = -\frac{\i}{2} \, h \sum_l \psi(\theta_k,\theta_l)^* \frac{\psi(\theta_{k+1},\theta_l)-\psi(\theta_{k-1},\theta_l)}{2h} + \mathrm{c.c.},
\end{align*}
we can define a constrained-search functional as well as linear pairings between $\rho$ and $u = v + A^2/2$ as well as $j$ and $A$,
\begin{align*}
    \langle u, \rho \rangle & = h \sum_k u_k \, \rho_k,
           \\
    \langle A, j \rangle & = h \sum_k A_k \, j_k.
\end{align*}
The densities and potentials may all be regarded as vectors in $\mathbb{R}^{N_G}$. Because all norms in finite dimensions are mathematically equivalent, we can choose to endow all function spaces with the same Euclidean $l^2(N_G)$ norm without losing compatibility. However, the norms may not be \emph{numerically} equivalent. Moreover, to connect to the continuum limit when $N_G \to \infty$, it is likely that one needs more carefully chosen norms. This is left for future studies. \newline
\indent The grid discretization makes it trivial to construct compatible finite-dimensional function spaces. However, this is not true in arbitrary basis expansions of density pairs and potential pairs. Unless the respective basis sets have special properties, compatibility is in general lost. \newline
\indent By solving the discretized Schr\"odinger equation, we then obtain the ground-state energy and the regularized energy
\begin{equation*}
    \bar{E}^{\lambda}_{\eps}(u,A) = \bar{E}^{\lambda}(u,A) - \frac{\eps}{2} \|u\|_2^2 - \frac{\eps}{2} \|A\|_2^2.
\end{equation*}
The universal density functional can be computed from the Lieb variational principle
\begin{equation}
    \label{eq:LiebVarQRing}
    F^{\lambda}_{\eps}(\rho,j) = \sup_{u,A} G^{\lambda}_{\eps}(u,A;\rho,j),
\end{equation}
with $G^{\lambda}_{\eps}(u,A;\rho,j) = \bar{E}^{\lambda}_{\eps}(u,A) - \langle u, \rho \rangle - \langle A, j \rangle$. We have found that a cutting-plane bundle method for convex optimization~\cite{CHENEY_NM1_253,KELLEY_JSIAM8_703} provides robust, though occasionally very slow convergence to the maximum value. In more detail, our implemented method maintains a ``bundle'' of data $(g_l,u_l,A_l,\sigma_l,k_l)$ from previous iterations. The bundle contains the function value $g_l = G^{\lambda}_{\eps}(u_l,A_l;\rho,j)$ and a supergradient $(\sigma_l,k_l)$ evaluated at $(u_l,A_l)$. Then a model function is defined by all the tangent planes encoded in the bundle,
\begin{equation}
    Q_l(u,A) = g_l + \langle u - u_l, \sigma_l \rangle +  \langle A - A_l, k_l \rangle.
\end{equation}
The next sample point $(u_{l+1},A_{l+1})$ is determined by maximizing the model function subject to a trust region constraint (to guard against $Q_l$ being unbounded, as may happen in the first iterations). \newline
\indent The stopping criterion requires care. As mentioned above, because the Moreau--Yosida regularization is only applied to $F^{\lambda}$, the energy functional $\bar{E}^{\lambda}_{\eps}$ is not more differentiable than the original $\bar{E}^{\lambda}$. Hence, there is no guarantee that $G^{\lambda}_{\eps}$ is differentiable with respect to the potentials at the maximum. This is particularly true for the KS potentials at $\lambda=0$, where $N$-representability constraints become more visible. Hence, it is not feasible to rely on vanishing (super)gradients as a stopping criterion for the optimization. This is connected to ground-state degeneracy and can be diagnosed by computing the energy gap to the first excited state.

\subsection{Kohn--Sham potentials from the Lieb variational principle}

We consider a discretization with $N_G=30$ grid points and set the electron-electron interaction to
\begin{equation*}
    W(\theta_1,\theta_2) = 3 \sqrt{1 + \cos(\theta_1-\theta_2)}.
\end{equation*}
We choose the external potentials
\begin{align*}
    v_\mathrm{ext}(\theta) & = \cos(\theta), \\
    A_\mathrm{ext}(\theta) & = 0.6,
\end{align*}
in order to obtain a non-trivial example that is nonetheless simple to specify. The external potentials are visualized in Fig.~\ref{fig:VcosUnregPot}.
The resulting Hamiltonian $H^{1}(v_\mathrm{ext},A_\mathrm{ext})$ has a highly correlated ground state with densities $(\rho,j)$ displayed in Fig.~\ref{fig:VcosUnregDens}. Performing maximization in the Lieb variation principle (Eq.~\eqref{eq:F-LF} or \eqref{eq:LiebVarQRing}) defining $F_{\eps=0}^{\lambda}(\rho,j)$ yields KS potentials $(u_\mathrm{KS},A_\mathrm{KS})$ as a by-product, visualized in Fig.~\ref{fig:VcosUnregPot}. The Hamiltonian $H^0(v_\mathrm{KS},A_\mathrm{KS})$, with $v_\mathrm{KS} = u_\mathrm{KS} - (A_\mathrm{KS})^2/2$, has a two-fold ground-state degeneracy and one of these ground-state densities is shown in Fig.~\ref{fig:VcosUnregDens}.
The vanishing gap is seen in Fig.~\ref{fig:VcosUnregEkink} and results in a non-differentiable kink in the ground-state energy $\bar{E}_{\eps=0}^0(u_\mathrm{KS},A_\mathrm{KS})$. The interacting density pair $(\rho,j)$ is a supergradient at this non-differentiable point, but it is neither a left- nor a right-derivative. Due to the limitation that our implementation is limited to pure states, and furthermore that the the choice of degenerate eigenvector basis is not optimized, it is seen in Fig.~\ref{fig:VcosUnregDens} that the interacting ground-state density pair $(\rho,j)$ is not reproduced exactly by the KS ground state. In general, exact reproduction requires mixed states.

\begin{figure}
\includegraphics[width=0.96\linewidth]{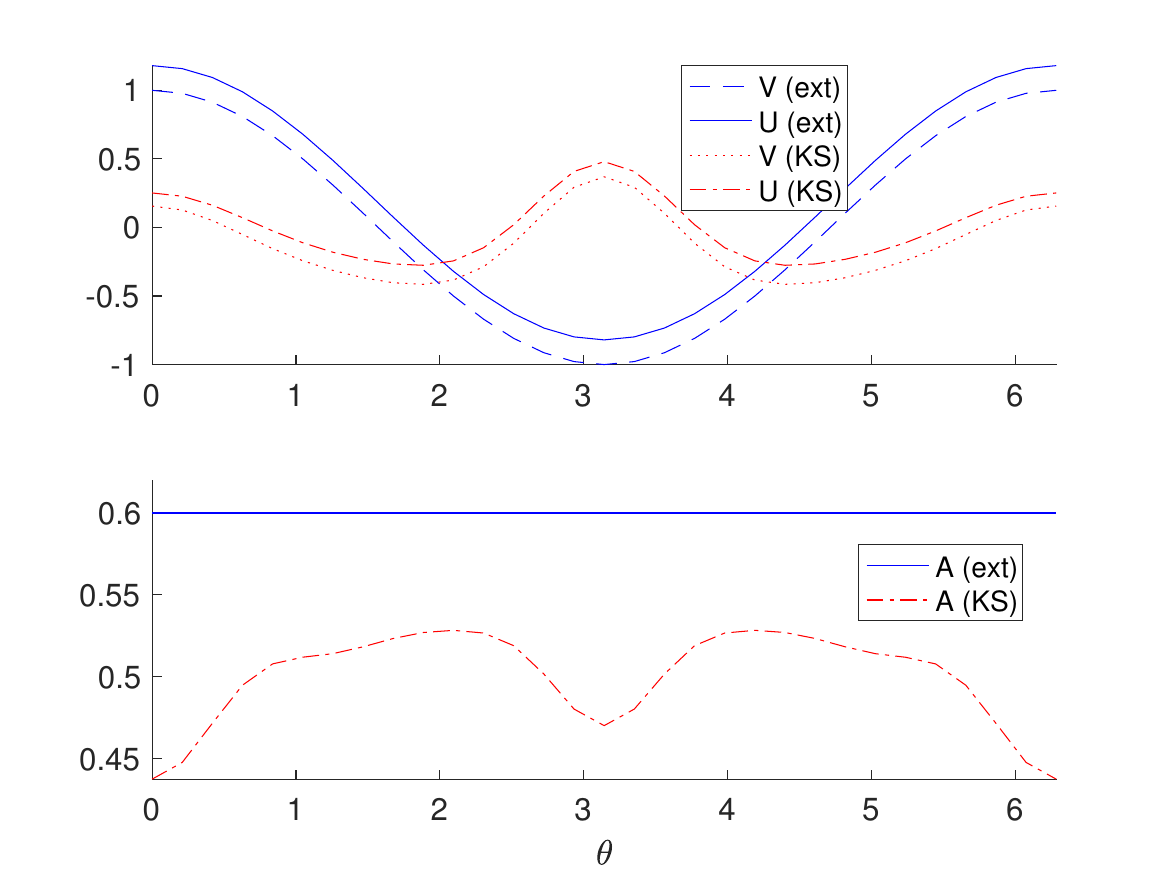}
\caption{External potentials and KS potentials in the unregularized case ($\eps=0$).}
\label{fig:VcosUnregPot}
\end{figure}

\begin{figure}
\includegraphics[width=0.96\linewidth]{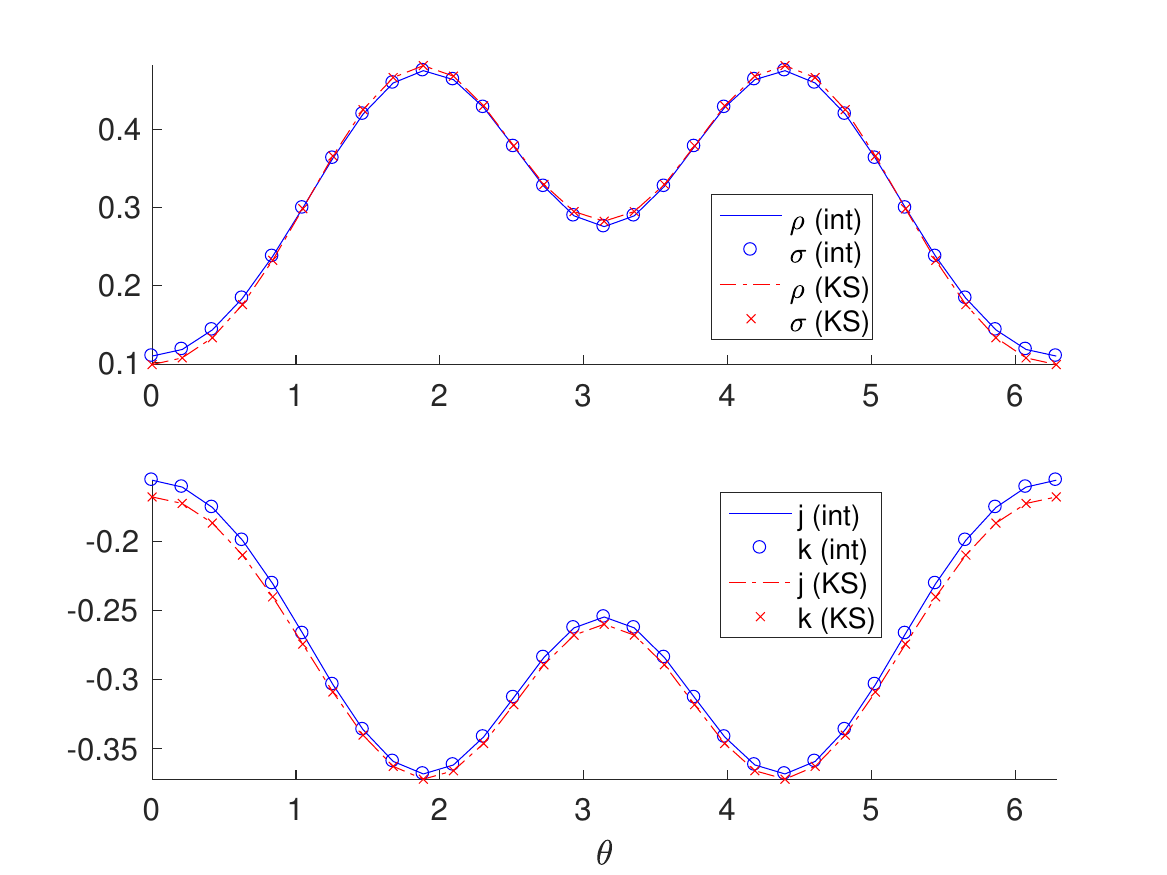}
\caption{The unregularized ground-state density pair $(\rho,j)$ for the correlated, interacting system subject to the external potentials together with the density pair $(\rho_\mathrm{KS},j_\mathrm{KS})$ for the uncorrelated KS system. Because $\eps=0$, the regularized density pair $(\sigma=\rho-\eps u_\mathrm{ext}, k=j-\eps A_\mathrm{ext})$ trivially coincides with unregularized density pair. Note that $(\rho,j)$ are nearly reproduced by the KS density pair $(\rho_\mathrm{KS}, j_\mathrm{KS})$, but failure of non-interacting $N$-representability prevents an exact match.}
\label{fig:VcosUnregDens}
\end{figure}

\begin{figure}
\includegraphics[width=0.96\linewidth]{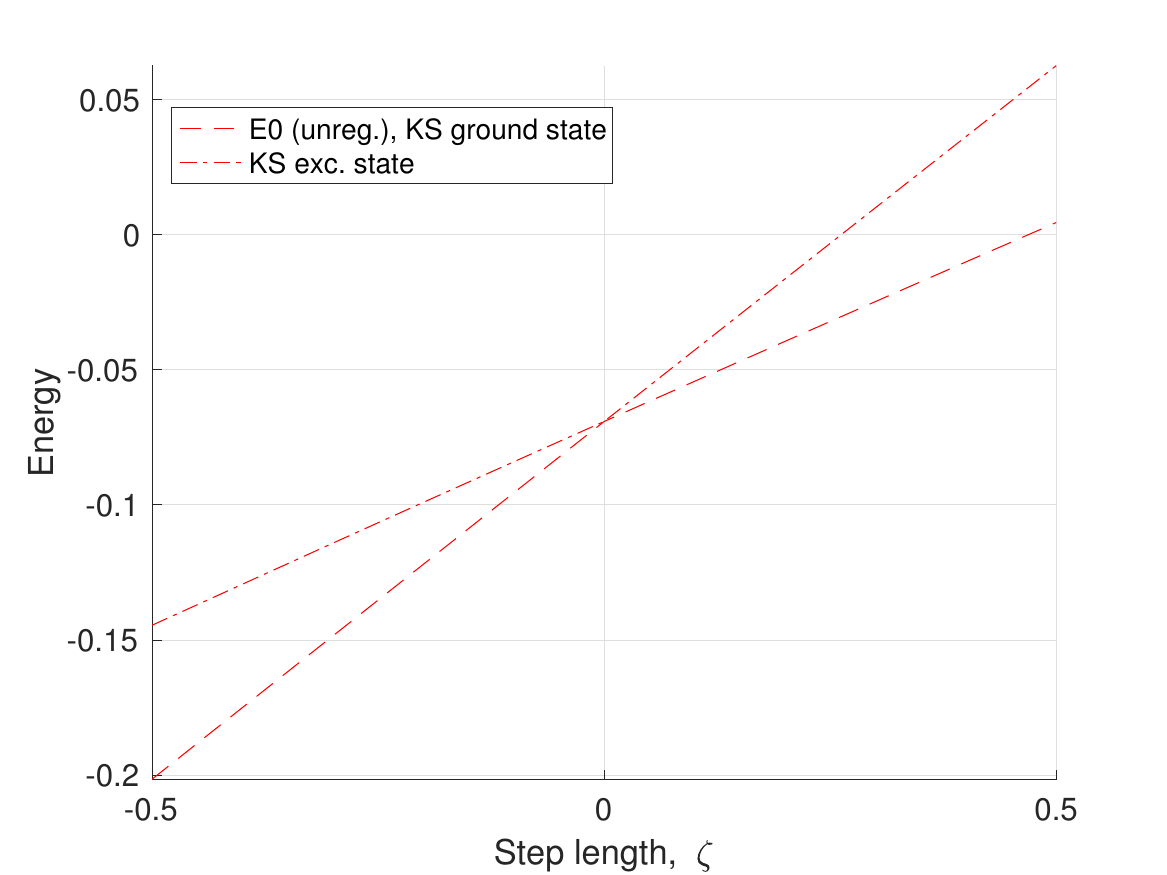}
\caption{The unregularized KS ground-state energy $\bar{E}^{0}_{\eps}(u_\mathrm{KS}+\zeta h\sigma,A_\mathrm{KS}+\zeta hk)$ and the first excited state as a function of the step length $\zeta$. The non-differentiable kink at $\zeta=0$ arises from a level crossing and the reference density pair from the interacting system corresponds to a particular supergradient at this kink.}
\label{fig:VcosUnregEkink}
\end{figure}

Next we illustrate the regularized setting by taking $\eps = 0.1$. This relatively large regularization parameter is used to make the effects of regularization noticeable. It is now the pair $(\sigma = \rho-\eps u_\mathrm{ext}, k = j - \eps A_\mathrm{ext})$ that takes over %most of 
the role played by the density pair in the unregularized setting. In particular, the Lieb variation principle now yields a KS potential pair $(u_\mathrm{KS},A_\mathrm{KS})$ such that $(\sigma_\mathrm{KS} = \rho_\mathrm{KS} - \eps u_\mathrm{KS}, k_\mathrm{KS} = j_\mathrm{KS} - \eps A_\mathrm{KS})$ coincides with the density pair $(\sigma,k)$, but $(\rho_\mathrm{KS},j_\mathrm{KS}) \neq (\rho,j)$. Hence, the KS potentials shown in Fig.~\ref{fig:VcosRegPot} are different from those in the unregularized setting (Fig.~\ref{fig:VcosUnregPot}). The resulting densities are shown in Fig.~\ref{fig:VcosRegDens}.

\begin{figure}
\includegraphics[width=0.96\linewidth]{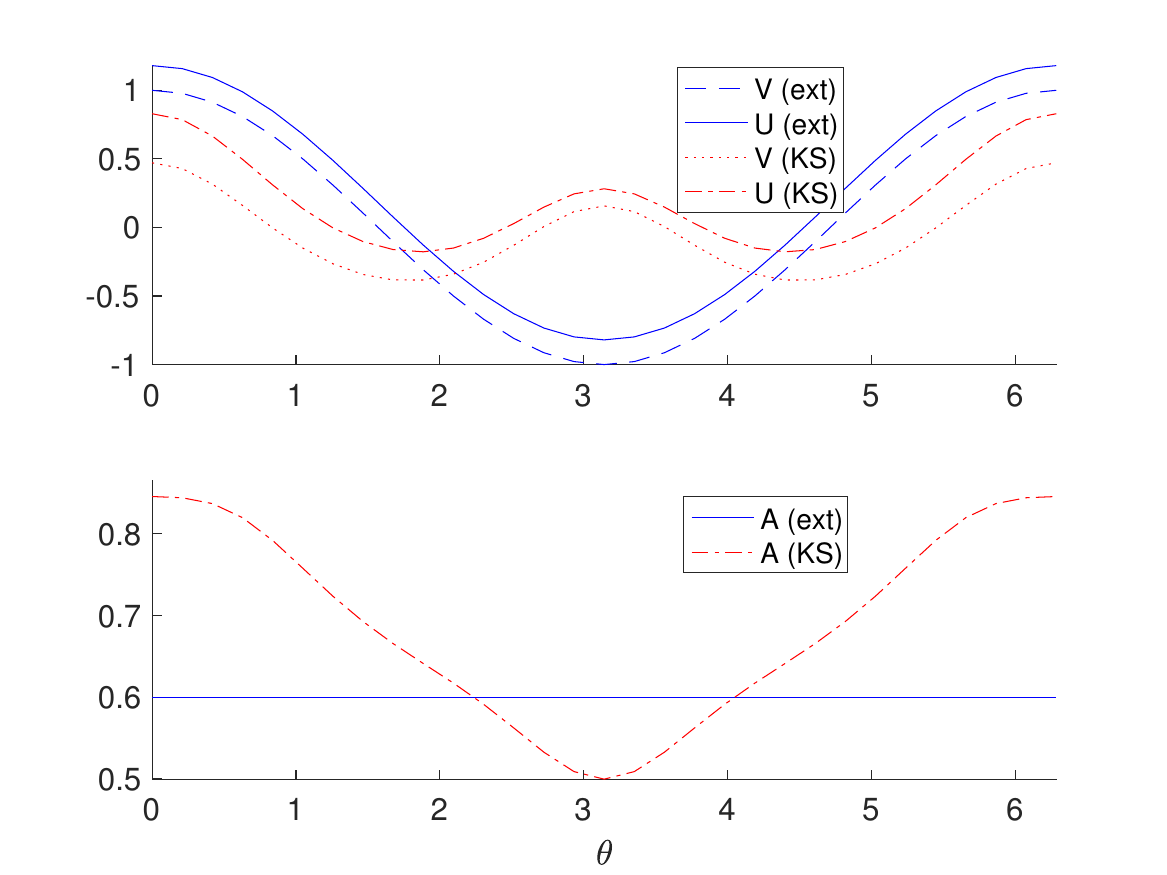}
\caption{External potentials and KS potentials in the case of Moreau--Yosida regularization with parameter value $\eps=0.1$.}
\label{fig:VcosRegPot}
\end{figure}

\begin{figure}
\includegraphics[width=0.96\linewidth]{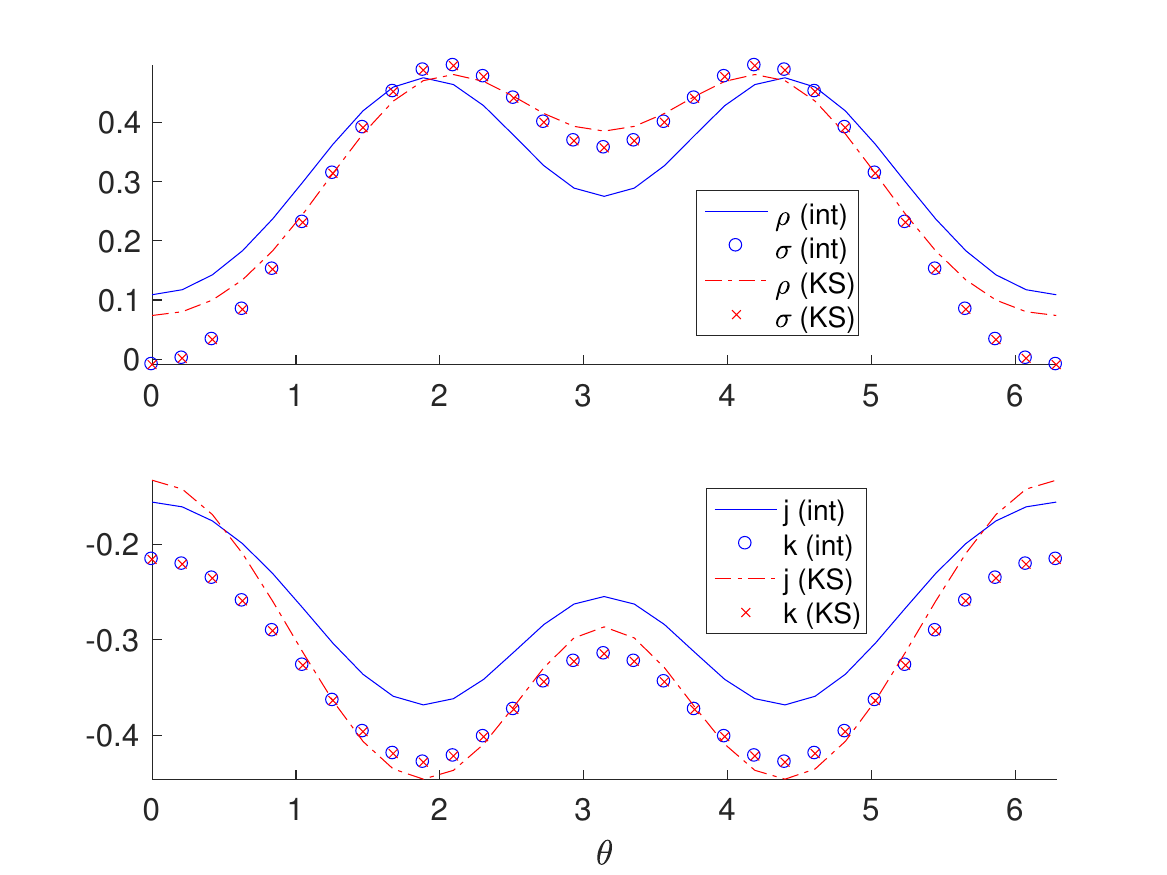}
\caption{The ground-state density pair $(\rho,j)$ for the correlated, interacting system subject to the external potentials together with the density pair $(\rho_\mathrm{KS},j_\mathrm{KS})$ for the uncorrelated KS system. The regularized density pair $(\sigma=\rho-\eps u_\mathrm{ext}, k=j-\eps A_\mathrm{ext})$ is very nearly reproduced by the regularized KS density pair $(\sigma_\mathrm{KS}=\rho_\mathrm{KS}-\eps u_\mathrm{KS}, k_\mathrm{KS}=j_\mathrm{KS}-\eps A_\mathrm{KS})$, but failure of non-interacting $N$-representability prevents an exact match.}
\label{fig:VcosRegDens}
\end{figure}

\subsection{Kohn--Sham potentials from the iterative algorithm}

In the previous section, the KS potentials were determined by first solving for correlated ground-state wave function of the interacting system, then constructing its densities, and finally plugging these densities into the Lieb variation principle. The iterative KS algorithm discussed in Sec.~\ref{sec:KS} above provides an alternative that does not require any \emph{a priori} information of the correlated ground state or its associated density. For simplicity, we have implemented the pure-state version of this algorithm, enabling us to see the consequences when a density pair is not representable by a pure ground state. The linesearch for the interpolation parameter $t$ was implemented in the following way:

\begin{itemize}
\item[(i)] Successively try $t = 1, \tfrac{1}{2}, \tfrac{1}{4}, \tfrac{1}{8}, \ldots$, until the criterion from the optimal damping step Eq.~\eqref{eq:oda-step} is fulfilled.

\item[(ii)] If already $t=1$ fulfills Eq.~\eqref{eq:oda-step}, then use this value. Otherwise, let $t=2^{-k}$ be the first parameter value such that Eq.~\eqref{eq:oda-step} holds and estimate the critical $t$ value by linear interpolation between $t=2^{-k+1}$ and $2^{-k}$. If the criterion is still not fulfilled at this $t$, perform another linear interpolation and choose the best of the sampled values.
\end{itemize}

The computation of gradients $\nabla F^1_\eps(\rho,j)$ is done using the Lieb variation principle, with a maximum of 300 bundle optimization iterations and a convergence criterion of $10^{-5}$ for stopping earlier. When there is a degenerate ground state, the gradient criterion does not apply and we instead test for a small gap and stagnated bundle iterations. In cases of numerically very small, but non-zero gap between the ground state and first excited state, our implementation may fail to obtain an adequate solution from the (pure-state) Lieb variation principle. As the algorithm was not formulated to account for such failures, the energy seen in the KS iterations may not be bounded from below by the true energy $\bar{E}^1(u_\mathrm{ext},A_\mathrm{ext})$, unless we override the reference density pair and instead use the actual density pair returned from the Lieb optimization. \newline
\indent Continuing with the same numerical example as in the previous section, we ran the KS iteration for different values of the regularization parameter. In the unregularized case, the consequences of failure of pure-state representability, both for the KS system and interacting systems corresponding to trial densities encountered in the course of the iterations, prevented a meaningful result. With Moreau--Yosida regularization, we were able to converge within the expected accuracy, given the finite precision of our implementation of the Lieb variation principle. In Fig.~\ref{fig:VcosKSenergy} the convergence of the energy difference,
\begin{align*}
    \Delta E_i  &= F^1_{\eps}(\rho_i,j_i) + \langle u_\mathrm{ext},\rho_i\rangle \\
    &\quad + \langle A_\mathrm{ext},j_i\rangle - \bar{E}^1_{\eps}(u_\mathrm{ext},A_\mathrm{ext}),
\end{align*}
is shown for four different values, $\eps=0.05, 0.1, 0.2, 0.3$, of the regularization parameter. Fig.~\ref{fig:VcosKSgradnorm} shows the convergence of the gradient norm,
\begin{equation*}
  \left\| (u_\mathrm{ext}, A_\mathrm{ext}) + \nabla F^1_{\eps}(\rho_i,j_i) \right\|_2,
\end{equation*}
which vanishes when the ground-state density of the interacting system has been reproduced. Although not encountered in the example studied here, small numerical inaccuracies especially in the Lieb variation principle lead to occasional small increases of the energy. The convergence is slow compared to experience with standard algorithms, such as Pulay's DIIS~\cite{PULAY_JCC3_556}, and approximate density functionals, as these result in quadratic convergence in favorable cases. However, most standard algorithms also lack formal convergence guarantees and have, for practical reasons, never been tested with the exact functional. 
An exception is the work by Wagner \emph{et al.}\ that did explore convergence of the exact functional using an algorithm applied to one-dimensional systems~\cite{Wagner2013,WagnerFollowUp}. In Ref.~\citenum{WagnerFollowUp} an adaptive choice of the damping (mixing) parameter was investigated, including discussions on line search and Hermite spline fit to the energy as a function of the damping parameter. (It is interesting to note that they use the curvature of the energy as information. In the regularized setting where derivatives are guaranteed to exist, the curvature is a key ingredient in the convergence proof of Ref.~\citenum{penz2019guaranteed}.) Furthermore, their study of an optimal damping parameter demonstrated numerically that convergence is more difficult for strongly correlated systems.\newline 
\indent The present work is the first time a KS vector potential, corresponding to an exact CDFT functional, is calculated using a KS iteration scheme. As expected, the convergence in Figs.~\ref{fig:VcosKSenergy} and \ref{fig:VcosKSgradnorm} is faster for larger values of the regularization parameter. This is partly due to the fact that the unregularized case features a KS system with vanishing gap and partly due to the increased regularity of the problem for larger $\eps$. 
\begin{figure}
\includegraphics[width=0.96\linewidth]{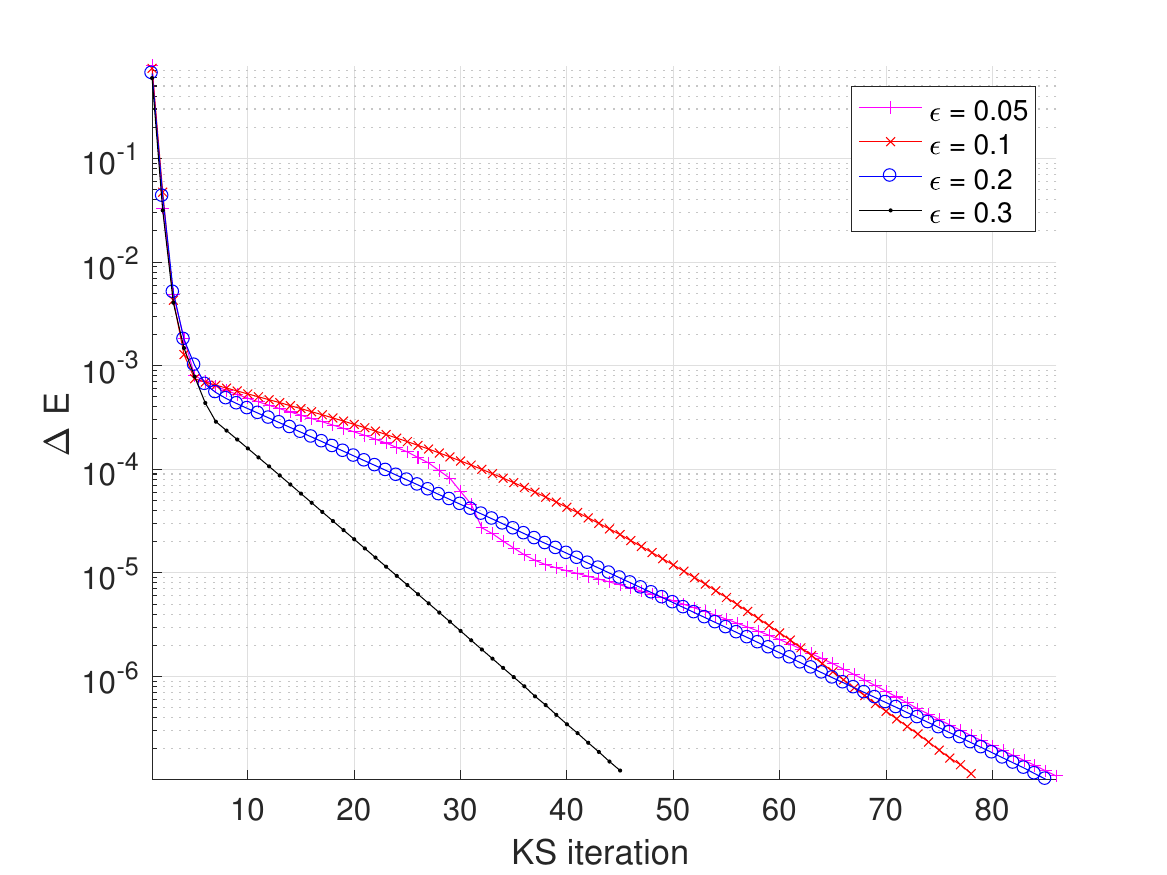}
\caption{Energy convergence of KS iterations for different values of the regularization parameter.}
\label{fig:VcosKSenergy}
\end{figure}

\begin{figure}
\includegraphics[width=0.96\linewidth]{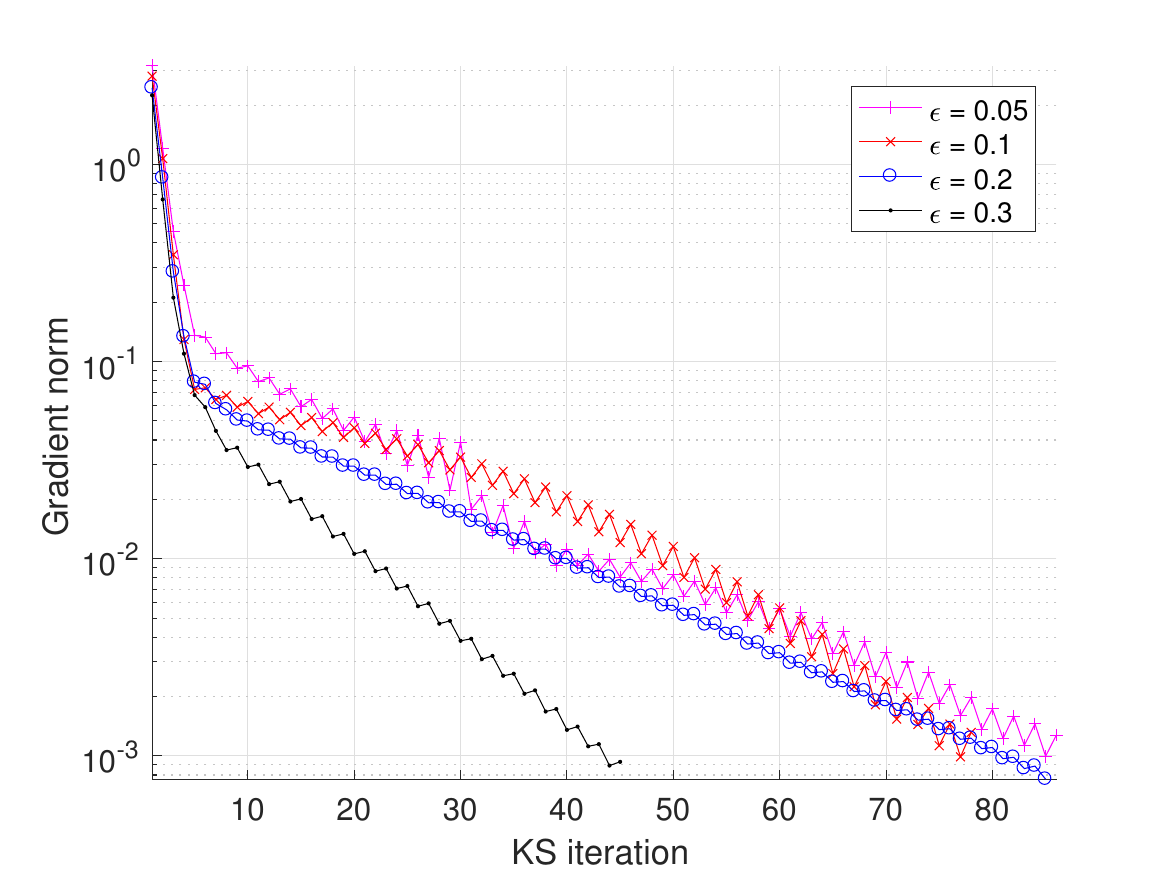}
\caption{Convergence of the gradient norm in the KS iterations for different values of the regularization parameter.}
\label{fig:VcosKSgradnorm}
\end{figure}

\begin{figure}
\includegraphics[width=0.96\linewidth]{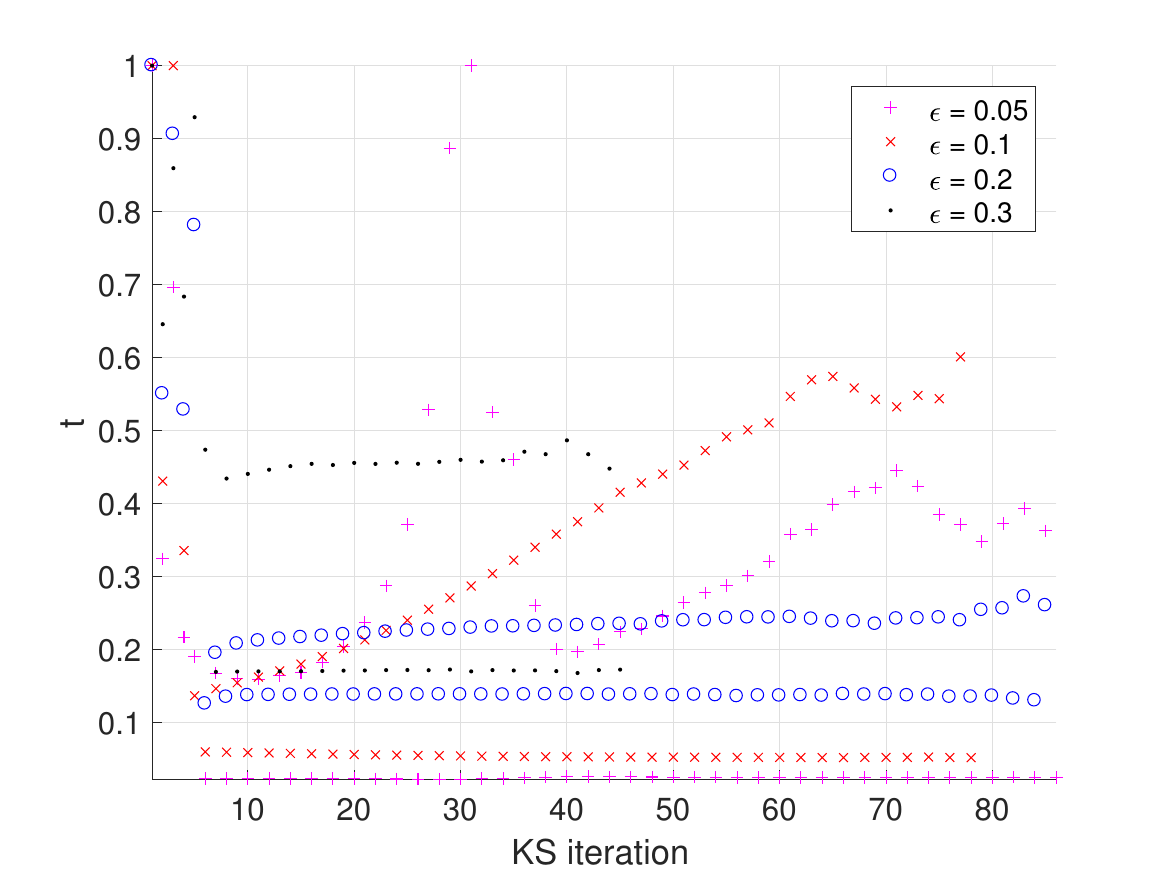}
\caption{Calculated linesearch step $t$ in the KS iterations.}
\label{fig:VcosKStparam}
\end{figure}

\begin{figure}
\includegraphics[width=0.96\linewidth]{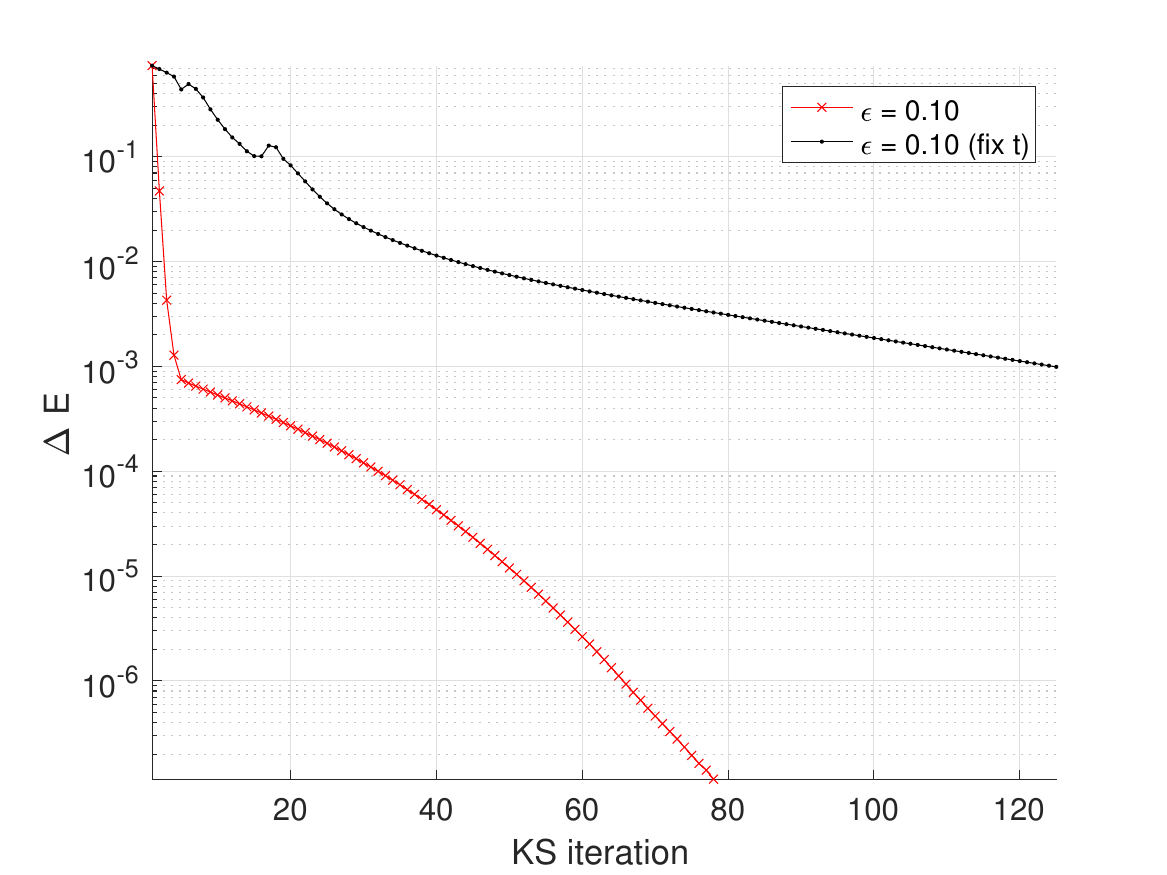}
\caption{Comparison of KS iterations with maximal $t$ determined in each iteration and with fixed $t=0.05$.}
\label{fig:ConstTEnergySeq}
\end{figure}

Finally, Fig.~\ref{fig:VcosKStparam} shows the calculated values of the damping parameter $t$ as a function of iteration number. The parameter values vary substantially between the examples with different $\eps$ and also from one iteration to the next. In particular, for $\eps = 0.1$ the MYKSODA iterations alternate between smaller $t \sim 0.05$ and larger values in the range $0.1 < t < 0.65$. A simpler iterative algorithm could use a fixed $t$ in all iterations, as was done in Ref.~\citenum{Wagner2013}. To explore this possibility, we fixed a conservative value $t=0.05$ for the damping parameter. As seen in Fig.~\ref{fig:ConstTEnergySeq} this yields dramatically slower convergence, showing that $t$ in general needs be chosen adaptively.

\section{Conclusions}

We have given a comprehensive account of the rigorous formulation of Kohn--Sham theory for CDFT. An important point is that textbook treatments of DFT rely on ill-defined functional derivatives~\cite{Lammert2007}. However, recent work has demonstrated that functional derivatives can be made well-defined and rigorous using Moreau--Yosida regularization~\cite{Kvaal2014,KSpaper2018}. We have extended that approach to functional differentiation in CDFT, enabling us to obtain well-defined Kohn--Sham potentials as well as an iteration scheme (MYKSODA). The presented MYKSODA is an algorithm for practical calculations in the setting of ground-state CDFT within a regularized framework. A toy model in the form of a quantum ring is solved numerically, and allowed a study of MYKSODA for the exact universal density functional. The calculations
illustrate the performance of the algorithm and highlight the difference to iteration schemes with a constant damping factor. It is also the first implementation of a Moreau--Yosida regularized Kohn--Sham approach. \newline
\indent While our model was solved numerically with the exact functional, this is of course not feasible for more realistic settings where we must resort to density-functional approximations. This raises the question of how to develop such approximations for the Moreau--Yosida regularized setting, or alternatively, of how to compute the Moreau--Yosida regularization of well-established density-functional approximations. This is an interesting topic for future investigation. \newline
\indent Central to the theory developed here was the concept of compatibility of spaces of densities and current densities. It allows a fully convex formulation of the theory and demands the use of Banach spaces for the basic variables. The respective $\vec L^p$ constraints for current densities were determined optimally in order to complement knowledge from traditional DFT and previous work on CDFT. This article sets the stage for further inquiries into the field, such as the possible full convergence of the iteration scheme and the study of approximate (regularized) functionals for CDFT.

\section*{Acknowledgments}
We thank an anonymous referee for improvements on our proof of Lemma~\ref{lemma:Efiniteness}.
This work was supported by the Norwegian Research Council through the
CoE Hylleraas Centre for Quantum Molecular Sciences Grant No.\ 262695. 
AL is grateful for the hospitality received 
at the Max Planck Institute for the Structure and Dynamics of Matter in Hamburg, while visiting  MP and MR.
MP acknowledges support by the Erwin Schr\"odinger Fellowship J 4107-N27 of the FWF (Austrian Science Fund) and is thankful for an invitation to the Hylleraas Centre just taking place writing this.
AL and SK were supported by ERC-STG-2014 under grant agreement No.~639508.
EIT was supported by the Norwegian Research Council through Grant No.~240674. 

\appendix
\section{A theorem on everywhere defined functionals on spaces of measurable functions}
\label{app:functional}

On any infinite-dimensional Banach space (assuming the axiom of choice) there exist everywhere defined linear maps that are unbounded. The following theorem shows that this cannot happen for linear functionals on spaces of measurable functions that are defined as integrals. The proof is based on a construction by D.~Fischer~\cite{stackexchange}. 
\begin{theorem}\label{thm:functionals}
    Let $B$ be a Banach space consisting of measurable functions $f : \mathbb{R}^n\to\mathbb{R}$. Let $g$ be a  measurable function. Then the functional $T : f \mapsto \int gf \d\mu$ is in $B^*$ if and only if for all $f \in B$, \[ \left| \int fg \; \d\mu \right| < +\infty. \]
\end{theorem}

\begin{proof}
Since a bounded linear functional must be everywhere defined, the only if part is trivial.
Suppose $g$ is measurable and that the integral $\int fg \,\d\mu$ exists for all $f \in B$. For $n \in \mathbb{N}$, define a sequence of bounded functions with bounded support,
\[ g_n(x) = \begin{cases} \qquad 0, & \lVert x\rVert > n, \\ \quad\;\; g(x), & \lVert x\rVert \leqslant n \text{ and } \lvert g(x)\rvert \leqslant n, \\ \frac{n}{\lvert g(x)\rvert}\cdot g(x), & \lVert x\rVert \leqslant n < \lvert g(x)\rvert.\end{cases} \]
Then $g_n$ is measurable for all $n$, and $h_n(x) = g_n(x)f(x) \to h(x) = g(x)f(x)$ for all $x$. Moreover $|g_n(x) f(x)| \leq |h(x)|$ for all $x$, the latter function being integrable by assumption. By the dominated convergence theorem,
\[\int_{\mathbb{R}^n} f(x)g_n(x)\d\mu \rightarrow\int_{\mathbb{R}^n} f(x) g(x)\d\mu,\]
as $n \to +\infty$. Thus, the family of continuous linear functionals $T_n \colon f \mapsto \int fg_n\d\mu$ is pointwise bounded.

The uniform boundedness principle states that a family $\{T_n\}$ of pointwise bounded linear functionals is in fact uniformly bounded. Thus, $\sup_n\|T_n\|_{B^*} < +\infty$. It then follows that 
\[\lVert T\rVert_{B^*}  \leq \sup_n  \|T_n\|_{B^*}  < +\infty. \]
Hence, $T \in B^*$. 
\end{proof}

%\vfill\null

\bibliography{refs}
\bibliographystyle{achemso}

\end{document}